\def\arxivversion{1} %
\ifx\arxivversion\undefined
    \documentclass[letterpaper, 10 pt, conference]{ieeeconf}  
\IEEEoverridecommandlockouts
\usepackage[sort,compress]{cite}
\usepackage[pdftex]{graphicx}
\usepackage{array}
\usepackage{url}
\usepackage{arydshln}
\usepackage{float}
\usepackage{scrextend}
\usepackage{blindtext}
\usepackage{amsmath,amsfonts,amssymb}
\usepackage{tikz}
\usepackage{xspace}
\usepackage{mathrsfs,mathtools}
\usepackage{mathalfa}
\usepackage{euscript}
\usepackage{csquotes}
\usepackage{datetime}
\usepackage{arydshln}
\usepackage{multirow}
\usepackage{blindtext}
\usepackage[caption=false]{subfig}
\usepackage[font=footnotesize]{caption}
\DeclareRobustCommand{\legendline}[1]{\hspace{-0pt}\tikz[#1,line width=0.4mm,baseline=-0.5ex]{\draw (0,0) -- (.2,0);}\hspace{-0pt}}

\definecolor{mblue}{rgb}{0,0.4470,0.7410}
\definecolor{morange}{rgb}{0.8500,0.3250,0.0980}
\definecolor{myellow}{rgb}{0.9290,0.6940,0.1250}
\definecolor{mpurple}{rgb}{0.4940,0.1840,0.5560}
\definecolor{mgreen}{rgb}{0.4660,0.6740,0.1880}
\definecolor{mcyan}{rgb}{0.3010,0.7450,0.9330}
\definecolor{mred}{rgb}{0.6350,0.0780,0.1840}
\definecolor{mgreenblue}{rgb}{0.0,1.0,0.5}
\definecolor{parulablue}{rgb}{0.2431,0.1490,0.6588}
\definecolor{parulalblue}{RGB}{39,151,235}
\definecolor{parulagreen}{RGB}{129,204,89}
\definecolor{parulayellow}{RGB}{249,251,21}
\definecolor{cblue}{rgb}{0,0.9,1}
\definecolor{corange}{rgb}{1,0.7,0}
\definecolor{mgray}{rgb}{0.8,0.8,0.8}

\newcommand{\mc}[1]{\mathcal{#1}}

\newcommand{\mr}[1]{\mathrm{#1}}
\newcommand{\mb}[1]{\mathbb{#1}}

\newcommand{\R}{\mathbb{R}}

\DeclareMathOperator*{\minimz}{minimize}
\DeclareMathOperator*{\subjto}{s.t.\ }

\newcommand{\norm}[1]{\left\lVert #1 \right\rVert}
\newcommand{\paren}[1]{\left( #1 \right)}

\newcommand{\minimize}[1]{\minimz_{#1}}

\newcommand{\lmin}{\lambda_{\min}}
\newcommand{\lmax}{\lambda_{\max}}
\newcommand{\cn}{\kappa}

\newcommand{\floor}[1]{{\left\lfloor #1 \right\rfloor}}

\newcommand{\Ecal}{\mathcal{E}}
\newcommand{\Fcal}{\mathcal{F}}

\newcommand{\Kcal}{\mathcal{K}}
\newcommand{\Lcal}{\mathcal{L}}
\newcommand{\Mcal}{\mathcal{M}}

\newcommand{\Xcal}{\mathcal{X}}
\newcommand{\Ucal}{\mathcal{U}}

\newcommand{\Tau}{\mathsf{T}}

\newcommand{\dx}{{n_\mr{x}}}
\newcommand{\du}{n_\mr{u}}

\newcommand{\DXF}[3]{D_XF\!\left(#1,#2\right)[#3]}
\newcommand{\Sig}{\Sigma_\tau}

\newcommand{\dt}{[t]}
\newcommand{\xdt}{x_{\dt}}

\usepackage{amsthm}
\theoremstyle{definition}
\newtheorem{definition}{Definition}

\theoremstyle{plain}
\newtheorem{theorem}{Theorem}
\newtheorem{lemma}{Lemma}
\newtheorem{proposition}{Proposition}

\newtheorem{remark}{Remark}
\theoremstyle{remark}

\renewcommand{\qedsymbol}{$\blacksquare$}

\let\labelindent\relax
\usepackage[inline]{enumitem}

\newcommand{\arxivver}[2]{%
  \ifx\arxivversion\undefined%
    #1%
  \else%
    #2%
  \fi%
}

\newcommand{\comment}[1]{}
\newcommand{\todo}[1]{{\color{red}#1}}
\newcommand{\remove}[1]{{\color{gray}#1}}

\newcommand{\citet}[1]{\cite{#1}}

\begin{document}
\title{A Quantitative Framework for Navigating Controller Design Tradeoffs under Computational Constraints}
\author{%
Chris Verhoek and Nikolai Matni
\thanks{C.~Verhoek and N.~Matni are with the Dept. of Electrical and Systems Engineering, University of Pennsylvania, United States. C.~Verhoek is also with the Control Systems Group, Eindhoven University of Technology, The Netherlands. Email addresses: \texttt{c.verhoek@tue.nl}, \texttt{nmatni@engineering.upenn.edu}.}%
\thanks{This work has been supported by AFOSR (FA9550261B038) and the European project COVER under the grant No.~101086228.}%
}

\maketitle

\begin{abstract}
	Computational constraints permeate the controller design process, and yet are rarely treated as explicit design constraints. Towards addressing this gap, we propose a quantitative framework that captures the effects of common design approximations, such as model order reduction, temporal discretization, horizon truncation, and solver accuracy, on both controller performance and computational requirements.  Our framework highlights that these approximations are \emph{tunable parameters} within an overall controller design process.  By leveraging incremental input-to-state stability, we show that bounding the aggregate effects of these approximations reduces to verifying a design-dependent sector bound on the difference between the deployed policy and an idealized baseline, with stability enforced via a small-gain condition. We operationalize these insights via a \emph{Design Meta-Problem} in which the performance gap is minimized subject to stability, real-time compute, and timing constraints. Finally, we instantiate the framework on a receding horizon LQR case study, and demonstrate a principled near-optimal navigation of tradeoffs among sampling rate, model order, horizon length, and solver iterations. 

\end{abstract}

\section{Introduction}

Controllers implemented on digital hardware must respect hard timing bounds imposed by plant dynamics, directly coupling achievable performance to available processing speed. Yet computational limits are rarely treated as first-class design constraints. Instead, tractability is achieved through a sequence of standard simplifications falling into roughly three categories: (i) modeling approximations, e.g., discretization of continuous-time dynamics, truncation of infinite horizons, and model order reduction, (ii) problem relaxations, e.g., convex relaxations of nonconvex optimal control problems, and early termination of iterative solvers, and (iii) architectural decompositions, e.g., layered control hierarchies that distribute computation across decision-making, planning, and regulation layers. 

While each simplification is individually well-motivated, they are typically introduced in isolation, and their combined effect on stability, performance, and computational feasibility is rarely analyzed jointly.  This leaves practitioners without principled guidance on which approximations to make, and at what fidelity. To address this gap, we propose a quantitative framework that makes computational constraints explicit within the control design process. Focusing on a single-layer optimal control problem, we treat model reduction, temporal discretization, horizon truncation, and solver accuracy as tunable design parameters, relating each directly to closed-loop performance degradation and real-time computational feasibility. This yields a {design meta-problem} in which performance is optimized subject to stability and real-time compute constraints. To our knowledge, this is the first framework for systematically and jointly navigating approximation choices under explicit computational constraints. %

\emph{Related work:\ }
The individual effects of design approximations (model reduction, horizon truncation, discretization, solver accuracy) are well understood, see~\citet{antoulas2005approximation, grune2008infinite, chen1995optimal, karapetyan2025closed} and references therein. What is lacking is a framework that treats these approximations as jointly tunable parameters under performance, stability, and computational constraints.  

Complementary lines of work address (i)~the division of compute resources between fast and slow control loops~\cite{rosolia2022unified, garg2021multi, stamouli2025layered},
which connects naturally to the broader research program on layered control architectures (LCAs)~\cite{matni2024quantitative}, (ii)~co-design and resource-aware methods~\cite{zardini2021codesign, lahijanian2018resource}, which take a system-level perspective in which the
controller is one component among many to be selected, and (iii)~event-triggered control~\cite{heemels2021event}, which addresses
computational constraints from a different angle, optimizing \emph{when} the controller acts rather than \emph{how} it computes. The most directly related body of work concerns compute-aware, real-time optimization for control, particularly for model-predictive control (MPC). Results in this space analyze performance loss incurred by early termination of online optimization in MPC~\cite{chen2025sampled, richter2011computational, rubagotti2014stabilizing, shitasos2025suboptimality, liao2020time, srikanthan2024closed, karapetyan2025closed} by leveraging (incremental) input-to-state stability. While important, these works treat model order, discretization, and horizon length as fixed rather than as tunable parameters.

\emph{Contributions:\ }
We propose a quantitative framework for optimizing controller performance subject to stability, real-time compute, and timing constraints. Our framework is centered around a \emph{Design Meta-Problem}~(DMP) in which performance degradation induced by design approximations, i.e., model reduction, temporal discretization, horizon truncation, and solver accuracy, is minimized subject to stability and computational feasibility. By leveraging incremental input-to-state stability~($\delta$ISS), we show that bounding degradation in performance reduces to verifying a design-dependent sector bound on the difference between deployed and baseline policies, with stability enforced via a small-gain condition. We operationalize the framework on a receding horizon linear quadratic regulator~(LQR) case study, demonstrating principled near-optimal navigation of tradeoffs among model order, sampling rate, horizon length, and solver iterations. %

\emph{Notation:\ }
$\mb{R}$ is the set of real numbers and let~$\R_+$ be the set of non-negative numbers. For a matrix $X\in\mb{R}^{n\times n}$, its smallest (largest) eigenvalue is denoted by~$\lmin(X)$ ($\lmax(X)$). The conditioning number of $X$ is denoted~$\cn(X)$, i.e., for~$X=X^\top$, $\cn(X)=\frac{\lmax(X)}{\lmin(X)}$. We denote the operator norm of~$X$ by~$\|X\|$. %
All proofs can be found \arxivver{in~\cite{extended}.}{in the appendices.}

\section{Problem Formulation}\label{s:problem}
We consider the nonlinear dynamical system
\begin{equation}\label{eq:system}
	\dot{x}_t = f_t(x_t,u_t),
\end{equation}
and associated infinite-horizon optimal control problem,%
\arxivver{%
\begin{subequations}\label{eq:base-problem}
\begin{align}
    \minimize{\pi} \quad & {\textstyle\int_{0}^\infty} c_t(x_t,u_t)\mathrm{d}t \\
     \subjto  \quad & \text{dynamics \eqref{eq:system}}, 
                        \ u_t = \pi_t(x_{[0,t]},u_{[0,t)}), \\
                        & x_t \in \Xcal_t,\, u_t \in \Ucal_t,\,x_0 = \xi,
\end{align}
\end{subequations}}{%
\begin{subequations}\label{eq:base-problem}
\begin{align}
    \minimize{\pi} \quad & \int_{0}^\infty c_t(x_t,u_t)\mathrm{d}t \\
     \subjto  \quad & \text{dynamics \eqref{eq:system}}, \\ & 
     u_t = \pi_t(x_{[0,t]},u_{[0,t)}), \\
     & x_t \in \Xcal_t,\, u_t \in \Ucal_t,\\
     & x_0 = \xi,
\end{align}
\end{subequations}
}
with state $x_t \in \R^{\dx}$, control input $u_t \in \R^{\du}$, dynamics model $f_t$, stage-wise cost function~$c_t$, state and input constraint sets $\Xcal_t$ and $\Ucal_t$, and state-feedback policy $\pi_t$. %
We denote\footnote{For clarity of notation, we suppress the $t$ subscript when referring to the policy $\pi=\{\pi_t\}_{t\geq0}$.} the cost achieved by a policy~$\pi$ in problem~\eqref{eq:base-problem} by~$J(\pi)$, and the optimal cost is given by $J^\star := J(\pi^\star)$, where~$\pi^\star$ is an optimal solution of~\eqref{eq:base-problem}.  We assume throughout that $J^\star$ is finite and is attained by a policy~$\pi^\star$.  %
The problem, as stated, is in general intractable to solve, and as such approximations to the problem must be introduced, motivated in part by computational constraints.  A goal of this paper is to bound the joint effects of these approximations on overall system performance, i.e., to bound $J(\hat \pi)-J(\pi^\star)$, for $\hat \pi$ the deployed policy. We highlight commonly used approximations next.

\subsection{Design Approximations}\label{ss:designapprox}
\begin{enumerate}[wide, leftmargin=0pt,label={\arxivver{\em}{\bf}\arabic*)},noitemsep,topsep=0pt]
    \item {\arxivver{\em}{\bf}Reduced Order Models and Policy Classes:} It is common to approximate the dynamics with a reduced order model~$\hat f$ from a model class~$\Fcal$ (e.g., linear, control-affine), and similarly to restrict our search to policies~$\hat\pi$ lying in a policy class~$\Pi$ (e.g., Markovian, time-invariant). We denote the joint reduced order model and policy class by~$\Mcal=\Fcal\times\Pi$. %
\item {\arxivver{\em}{\bf}Horizon Length:} Rather than solving problem~\eqref{eq:base-problem} over an infinite horizon, it is instead solved over a \emph{finite} horizon~$T$. The cost for $t\in(T,\infty)$ is captured by augmenting the finite-horizon objective with a terminal cost~$c_T(x_T)$.

\item {\arxivver{\em}{\bf}Temporal Discretization:} To make solving the optimal control problem~\eqref{eq:base-problem} computationally tractable, it is often temporally discretized with a sampling time~$\tau$.

\item {\arxivver{\em}{\bf}Optimization Iterations:} The resulting reduced-order, finite-horizon, discrete-time problem, is often deployed in a receding horizon fashion, which demands online re-solving at each step using iterative algorithms (e.g., gradient descent, interior point methods). Due to loop rate and compute constraints, only a finite number~$N$ of iterations can be run.
\end{enumerate}

\begin{remark}
	This list is not exhaustive, and the proposed framework can easily be extended to handle other bespoke approximations and design choices. Moreover, the \emph{order} in which these design approximations are introduced may depend on the control design problem at hand.
\end{remark}

While these approximations are necessary to make the control design
problem~\eqref{eq:base-problem} tractable, the resulting deployed 
policy~$\hat\pi$ must still respect certain key \emph{design 
constraints}.

\subsection{Design Constraints}\label{ss:designcons}
\begin{enumerate}[wide, leftmargin=0pt,label={\arxivver{\em}{\bf}\arabic*)},noitemsep,topsep=0pt]
    \item {\arxivver{\em}{\bf}Stability Constraints:} A minimal requirement is for the deployed policy~$\hat\pi\in\Pi$ to be stabilizing.  This, in turn, may impose lower bounds on model and policy class complexity~$\Mcal$, on loop rate and sampling time, on horizon length, and on the number of optimization iterations.
    \item {\arxivver{\em}{\bf}Performance Constraints:} These may be imposed from the design problem at hand or due to integration within a larger LCA.  These could include, e.g., tracking error or convergence rate guarantees.  These can be enforced by requiring a policy~$\hat\pi$ satisfying~$R(\hat\pi)\leq 0$, for a suitably defined requirements function~$R$.
    
\item {\arxivver{\em}{\bf}Offline Compute Constraints:} While offline computation can accommodate heavier computational demands,
we still require methods that can be solved within a prescribed amount of time, e.g., polynomial time synthesis algorithms.

\item {\arxivver{\em}{\bf}Online Compute Constraints:} Online computation must respect constraints as imposed by, e.g., loop rate constraints and hardware capabilities.
\end{enumerate}

As before, this list is not exhaustive, and our framework naturally accommodates additional constraints.

\subsection{Problem Statement} While the above design approximations and constraints are well understood individually, they have yet to be holistically quantified and optimized.  We therefore seek to provide a systematic framework for navigating the tradeoffs introduced by these design approximations and constraints. Toward that end, in the next section, we formulate a DMP for the optimal control problem~\eqref{eq:base-problem} that optimizes the performance gap of the deployed policy induced by the design approximations of~\S\ref{ss:designapprox} subject to the constraints of~\S\ref{ss:designcons}, with a particular emphasis on explicitly encoding \emph{computational constraints}.

\section{The Design Meta-Problem}\label{s:qframe}

To track the contribution of each approximation to the overall performance gap, we associate a dedicated approximate policy with each design parameter, wherein subscripts identify the relevant quantities that define the approximation. For example, $\hat\pi_\Mcal$ denotes the policy computed using the reduced order model and policy class~$\Mcal$, $\hat\pi_T$ denotes the policy computed for a finite horizon of length~$T$, $\hat\pi_\tau$ denotes the policy computed using a sampling time~$\tau$, and $\hat\pi_N$ denotes the policy computed using~$N$ compute iterations. There is an implicit dependence induced by the order of the approximations, e.g., $\hat\pi_N$ here is inherently dependent on the choices of $\tau,T,\Mcal$, i.e.,~$\hat\pi_{N,\tau,T,\Mcal}$. However, to reduce notational clutter, we suppress this implicit dependence in our notation. In our model, the deployed policy (on~$f$) is $\hat\pi_N$, and hence our goal is to understand the effect of the design choices on the performance gap~$\Ecal := J(\hat\pi_N) - J(\pi^\star)$.

\begin{remark}\label{rem:notoptpi}
    We use the optimal policy $\pi^\star$ as a natural 
    baseline, but this is not a requirement, and~$\hat\pi$ can be compared against any baseline policy. This is 
    particularly useful when $\pi^\star$ is impossible to compute, but 
    a performant yet non-implementable policy (e.g., one that violates 
    online compute constraints) is available as a baseline.
\end{remark}

Towards quantifying the tradeoffs that design choices induce in the overall control design problem, we take inspiration from work on tradeoffs in large-scale learning~\cite{bottou2007tradeoffs}, and consider the following \emph{{performance gap} decomposition}:
\arxivver{%
\begin{align}
    \Ecal &= J(\hat\pi_N) - J(\pi^\star) \nonumber\\
    &\begin{multlined}[b]
    	= \underbrace{J(\hat\pi_N) - J(\hat\pi_\tau)}_\text{Optimization Error $\Ecal_\textsc{opt}(N)$} + \underbrace{J(\hat\pi_\tau) - J(\hat\pi_{T})}_\text{Discretization Error $\Ecal_\textsc{dt}(\tau)$} \\   + \underbrace{J(\hat\pi_{T}) - J(\hat\pi_\Mcal)}_\text{Finite Horizon Error $\Ecal_\textsc{fh}(T)$} + \underbrace{J(\hat\pi_{\Mcal}) - J(\pi^\star)}_\text{ROM Error $\Ecal_\textsc{rom}(\Mcal)$} 	
 \end{multlined} \nonumber\\
    & =: \Ecal_\textsc{opt}(N) + \Ecal_\textsc{dt}(\tau) + \Ecal_\textsc{fh}(T)+ \Ecal_\textsc{rom}(\Mcal). \label{eq:app-error}
\end{align}}{%
\begin{align}
    \Ecal &= J(\hat\pi_N) - J(\pi^\star) \nonumber\\
    & = \underbrace{J(\hat\pi_N) - J(\hat\pi_\tau)}_\text{Optimization Error $\Ecal_\textsc{opt}(N)$} + \underbrace{J(\hat\pi_\tau) - J(\hat\pi_{T})}_\text{Discretization Error $\Ecal_\textsc{dt}(\tau)$}   + \underbrace{J(\hat\pi_{T}) - J(\hat\pi_\Mcal)}_\text{Finite Horizon Error $\Ecal_\textsc{fh}(T)$} + \underbrace{J(\hat\pi_{\Mcal}) - J(\pi^\star)}_\text{ROM Error $\Ecal_\textsc{rom}(\Mcal)$}  \nonumber\\
    & =: \Ecal_\textsc{opt}(N) + \Ecal_\textsc{dt}(\tau) + \Ecal_\textsc{fh}(T)+ \Ecal_\textsc{rom}(\Mcal). \label{eq:app-error}
\end{align}
}%
Decomposition~\eqref{eq:app-error} shows that the {performance gap} of the deployed policy~$\hat\pi_N$ relative to the baseline policy~$\pi^\star$ can be decomposed into error terms that depend on the choice of horizon~$T$ (i.e., $\Ecal_\textsc{fh}$), the sampling time~$\tau$ (i.e., $\Ecal_\textsc{dt}$), the choice of reduced order model class~$\Mcal$ (i.e., $\Ecal_\textsc{rom}$), and the number of optimization iterations~$N$ (i.e., $\Ecal_\textsc{opt}$).  We note that the order of the decomposition will typically depend on the design problem at hand, and that not all approximations need be introduced. %
Moreover, note that the error terms in~\eqref{eq:app-error} share the same implicit dependence on earlier approximation errors as the approximate policies do. %
As a final note, 
this sequential introduction of design approximations is common in every field in engineering, and is not only done to make the problem computationally tractable, but also to make solutions understandable and interpretable. Hence, our proposed framework not only leads to computational tractability, but also to \emph{intellectual tractability}: each error term has a clear operational interpretation, and the sequential structure mirrors how experienced engineers approach these problems.

Absent additional constraints, one would simply minimize all error terms simultaneously. However, as outlined in \S\ref{ss:designcons}, engineers are subject to several constraints when designing systems, which we encode in the following (informal) DMP:
\arxivver{\vspace{-1em}}{}
\begin{subequations}\label{eq:meta-problem}
\begin{align}
     \minimize{N,\,\tau,\,T,\,\Mcal} \quad &  \Ecal_\textsc{opt}(N) + \Ecal_\textsc{dt}(\tau) + \Ecal_\textsc{fh}(T) + \Ecal_\textsc{rom}(\Mcal) \label{eq:meta-obj}  \\
     \subjto \quad &  T\geq T_{\min}, \, \tau \leq \tau_{\max}, \label{eq:meta-timing}\\
                        & \exists(\hat{f},\hat\pi)\in\Mcal \text{ stable, s.t. } \hat\pi\text{ stabilizes \eqref{eq:system},} \label{eq:meta-stable}\\
                        &  \hat\pi_N \text{ is stabilizing,} \label{eq:meta-opt}\\
                        &  \Tau(N,\tau,T,\Mcal)\leq \tau. \label{eq:meta-compute}
\end{align}
\end{subequations}

In solving the DMP, we seek to minimize the objective function~\eqref{eq:meta-obj}, which characterizes the total error induced by design approximation parameters $(N,\tau,T,\Mcal)$, as defined in decomposition~\eqref{eq:app-error}, subject to practical design constraints.  Constraint~\eqref{eq:meta-timing} ensures that the horizon~$T$ is sufficiently large %
and the sampling time~$\tau$ is sufficiently small, %
so as to ensure stability of the resulting policy when deployed on the base problem~\eqref{eq:base-problem}.\footnote{For this informal problem definition, we assume that a minimum horizon $T_{\min}$ and a maximum sampling time $\tau_{\max}$ are provided.  In \S\ref{s:formalization}, we show how these can be implicitly characterized by a small-gain constraint.}
Constraint~\eqref{eq:meta-stable} imposes a minimal complexity bound on the reduced order model and policy class~$\Mcal$, such that a stabilizing controller for the base problem~\eqref{eq:base-problem} can be found using a reduced order model and policy defined by~$\Mcal$.
Constraint~\eqref{eq:meta-opt} ensures that enough optimization algorithm iterations~$N$ are run so that~$\hat\pi_N$ is stabilizing.\footnote{We omit the performance constraint $R(\hat{\pi}_N)\leq 0$ for 
simplicity, but note that if $R$ satisfies smoothness conditions 
analogous to those imposed on the cost $c$ in Lemma~\ref{lem:delta_J}, the same 
$\delta$ISS-based analysis applies to bound $R(\hat{\pi}_N)$.}

Finally, in constraint~\eqref{eq:meta-compute} we introduce the compute time~$\Tau$, which maps the design parameters~$(N,\tau,T,\Mcal)$ to a wall-clock time that is needed to compute the next control action defined by~$\hat\pi_N$. Motivated by receding horizon deployment of policies, we constrain the control action update time~$\Tau$ to be no larger than the sampling time~$\tau$. This constraint can be naturally adapted to other deployment paradigms, and we emphasize that this constraint captures the interplay between \emph{online compute and loop rate constraints}. %

The DMP~\eqref{eq:meta-problem} makes the tensions inherent in controller design under computational constraints explicit, which we qualitatively summarize in Table~\ref{tab:tradeoffs}.  The table shows the typical evolution of objective and constraint components (rows) when design variables (columns) increase: compute time invariably increases (decreases) when a performance metric decreases (increases), illustrating the unavoidable tradeoffs to be made in controller design. While conceptually useful, the DMP is not yet actionable, as the expressions appearing in~\eqref{eq:meta-problem} are difficult to quantify, let alone jointly optimize. We show next that assuming a $\delta$ISS baseline system reduces both the stability constraints and the performance gap bound to a design-dependent sector bound on $\|(\hat \pi - \pi^\star)(x)\|$, yielding a tractable reformulation of the DMP.

\begin{table}[t]
    \centering
    \caption{Typical evolution of the error terms in~$\Ecal$ and compute time~$\Tau$ when design parameters $(N,\tau,T,\Mcal)$ increase.}
    \label{tab:tradeoffs}
    \begin{tabular}{clcccc}
    \hline
                       &                                & $N$ & $1/\tau$  & $T$ & $\Mcal$ \\
         \hline
         $\Ecal_\textsc{opt}$ &(optimization error)     & $\searrow$ & $-$ & $-$ & $-$ \\
         $\Ecal_\textsc{dt}$  &(discretization error)   & $-$ & $\searrow$ & $-$ & $-$\\
         $\Ecal_\textsc{fh}$  &(finite horizon error)   & $-$ & $-$ & $\searrow$ & $-$\\
         $\Ecal_\textsc{rom}$ &(ROM error)              & $-$ & $-$ & $-$ & $\searrow$ \\
         $\Tau$ &(compute time)                  & $\nearrow$ & $\nearrow$ & $\nearrow$ & $\nearrow$\\
         \hline
    \end{tabular}
    \arxivver{\vspace{-1.5em}}{}
\end{table}

\section{Formalizing the Design Meta-Problem via Incremental Input-to-State Stability}\label{s:formalization}

Towards simplifying exposition, we assume that 
problem~\eqref{eq:base-problem} is time-invariant, i.e., $c_t\equiv c$, $f_t \equiv f$, $\Xcal_t\equiv \Xcal$, and $\Ucal_t \equiv \Ucal$, so that the optimal {baseline} policy $\pi^\star$ is Markovian and time-invariant. %
Although time-varying policies may be optimal for sub-problems induced by design approximations, we will consider the deployment of these solutions in a receding horizon fashion, and hence, we can assume without loss of generality that $\hat\pi_t(x_{[0,t]},u_{[0,t)})=\hat\pi(x_t)$.

To bound the performance gap of $\hat \pi$, we view it as a perturbed version of the baseline policy~$\pi^\star$, i.e.,~$\hat \pi(x) = \pi^\star(x) + \Delta(x)$, where $\Delta(x):=(\hat \pi - \pi^\star)(x)$ denotes the policy perturbation.  %
Towards that end, for a fixed policy $u = \pi(x)$, define the closed-loop dynamics $f^{\pi}(x,u) := f(x, \pi(x) + u)$, and $f^{\pi}(x) := f^{\pi}(x,0)$.  To characterize the difference in cost achieved by the {baseline} closed-loop system~$f^\star := f^{\pi^\star}$ and a closed-loop system evolving under {the deployed} policy~$\hat\pi$, i.e., $\hat f := f^{\hat \pi}$, we observe that $\hat f(\hat x) = f^\star(\hat x,\Delta(\hat x))$, i.e., it is \emph{an input perturbed version} of the {baseline} system~$f^\star$.  

\begin{remark}
    Throughout, we use $(x, x')$ and $(f, f')$ to denote generic 
    nominal and perturbed systems, reserving $(x^\star, \hat x)$ and 
    $(f^\star, \hat f)$ for the baseline and deployed systems 
    specifically. The former appear in the general results of 
    \S\ref{ss:smallgain}, while the latter 
    connect those to formalizing the DMP~\eqref{eq:meta-problem}.
\end{remark}

Incremental input-to-state stability, 
which characterizes a notion of robust stability between trajectories, provides a natural framework for bounding the effects of such \emph{input perturbations} on closed-loop behavior, see~\citet{angeli2002lyapunov} and references therein.  We show next that if the baseline system $f^\star$ is $\delta$ISS, stability and performance guarantees can then be characterized via a suitable small-gain condition defined in terms of a tunable sector bound on the policy difference.

\subsection{Preliminaries}

Recall definitions of standard comparison functions~\cite{khalil2002nonlinear}: 
\arxivver{}{\begin{definition}}
A function $\gamma: [0,a) \to\R_+$ is of class $\mathcal{K}$ if it is continuous, strictly increasing, and $\gamma(0) = 0$. It is of class $\mathcal{K}_\infty$ if it is of class $\Kcal$ with $a=\infty$ and, in addition, $\lim_{r\to\infty}\gamma(r)=\infty$.
\arxivver{}{\end{definition}\begin{definition}}
	A function $\beta : [0,a) \times \R_+ \to\R_+$ is of class $\mathcal{KL}$ if: \emph{i)}~it is continuous, \emph{ii)}~for each fixed~$t$, $\beta(\cdot, t)$ is class~$\mathcal{K}_\infty$, and \emph{iii)}~for each fixed~$x$, $\beta(x, \cdot)$ is strictly decreasing and $\lim_{t\to\infty}\beta(x,t)=0$.
\arxivver{}{\end{definition}}
We restrict ourselves to $\Kcal_\infty$ comparison functions, which corresponds to assuming global $\delta$ISS.\footnote{This can be relaxed to local $\delta$ISS on a compact domain under suitable assumptions on the input perturbations.}

\begin{definition}[Def. 4.1, \citet{angeli2002lyapunov}]\label{def: delta ISS}
     A system $\dot x = f(x,u)$ is %
     $\delta$ISS if there exists a class $\Kcal\Lcal$ function $\beta$ and a class $\Kcal_\infty$ function~$\gamma$, such that for all initial conditions $x_0, x_0' \in \R^\dx$, all input signals $\{u_t\}_{t \geq 0},\{u'_t\}_{t\geq 0}$, and all $t\geq 0$:
    \[
        \norm{x_t-x_t'} \leq \beta\paren{\norm{x_0 - x_0'}, t} + \gamma\big({\textstyle\sup_{0 \leq s \leq t}} \norm{u_s-u_s'}\big).
    \]
\end{definition}
Definition \ref{def: delta ISS} says that: \emph{i)}~trajectories generated by $\delta$ISS systems converge towards each other if they begin from different initial conditions, and \emph{ii)}~the effect of bounded differences in inputs $\{u_t-u'_t\}_{t\geq 0}$ on trajectories is bounded.
\emph{$\delta$ISS Lyapunov functions} can be used to certify whether a system satisfies the $\delta$ISS condition of Definition~\ref{def: delta ISS}.

\begin{definition}[Def. 5.1, \citet{angeli2002lyapunov}]\label{def:deltaISS}
	A smooth function $V(x, x') : \mathbb{R}^\dx \times \mathbb{R}^\dx \rightarrow \mathbb{R}_{+}$ is called a \emph{$\delta$ISS-Lyapunov function} if
\begin{equation}\label{eq: lyap comp fcns}
\alpha_1(\|x - x'\|) \leq V(x, x') \leq \alpha_2(\|x - x'\|)
\end{equation}
for some $\alpha_1, \alpha_2 \in \mathcal{K}_\infty$, and there exists $\kappa \in \mathcal{K}_\infty$ such that for any $u, u' \in \mathcal{U}$ and any $x, x' \in \mathbb{R}^\dx$, we have that $\kappa(\|x - x'\|) \geq \|u - u'\|$ implies
\begin{equation}\label{eq: vdot_0}
\dot V(x,x') = \tfrac{\partial V}{\partial x} f(x, u) + \tfrac{\partial V}{\partial x'} f(x', u') < -\varrho(\|x - x'\|),
\end{equation}
is satisfied with~$\varrho$ positive definite.
\end{definition}
It is shown in~\arxivver{\cite[Thm.~2]{angeli2002lyapunov}}{\citet[Thm.~2]{angeli2002lyapunov}} that when~$u,u'\in\Ucal$, with~$\Ucal$ a compact set, then the existence of a $\delta$ISS-Lyapunov function for a system $\dot x = f(x,u)$ is equivalent to it being $\delta$ISS. Further, condition~\eqref{eq: vdot_0} is implied by the existence of a $\delta$ISS Lyapunov function with Lie derivative~$\dot V$ satisfying the exponential dissipation inequality:
\begin{equation}\label{eq: diss}
    \dot V(x,x') \leq -cV(x,x')+ \sigma(\|u-u'\|),
\end{equation}
for some constant $c>0$ and class $\Kcal_\infty$ (supply) function~$\sigma$.

\subsection{A Small-Gain Robustness Criterion}\label{ss:smallgain}

We now %
analyze the evolution of a perturbed system
\begin{equation}\label{eq: pert sys}
    \dot x'=f(x',\Delta(x')+u)
\end{equation}
relative to an unperturbed system $\dot x = f(x,u)$.  We consider sector bounded perturbations $\Delta$ for which there exists an $L>0$, such that $\|\Delta(x')\|\leq L\|x'\|$ for all $x'\in\R^\dx$ \emph{that can be generated by the perturbed system~\eqref{eq: pert sys}}.  
Under a suitable \emph{robustness} assumption on the nominal system $\dot x = f(x,u)$, akin to a small-gain condition, we can guarantee that $\|x_t-x'_t\|$ is bounded, and converges to zero:

\begin{proposition}\label{prop: small-gain}
	Suppose that $\dot x = f(x,u)$ is $\delta$ISS with $\delta$ISS-Lyapunov function $V_\delta$ satisfying bound~\eqref{eq: lyap comp fcns} and inequality~\eqref{eq: diss}, characterized by class $\Kcal_\infty$ functions $\alpha_1,\alpha_2,\sigma$ and constant $c>0$. Let $\rho\in(0,1)$ and assume that
	\begin{equation}\label{eq:smallgaincondition}
		\sigma(2Lz) \leq \rho c\alpha_1(z), \quad \text{for all } z\geq 0,
	\end{equation}
	that $x'_0=x_0$, and that $u_t\equiv 0$ for all $t\geq0$.  Then,
	\begin{equation}\label{eq:small-gain}
	    \|x_t-x'_t\|\leq \alpha_1^{-1}\big(e^{-c(1-\rho)t}\alpha_2(\|x_0\|)\big).
	\end{equation}
\end{proposition}

Prop.~\ref{prop: small-gain} follows from a standard small-gain argument, but its implications are profound. In particular, ensuring stability of the perturbed system~\eqref{eq: pert sys} is reduced to verifying the sector bound $\|\Delta(x')\|\leq L\|x'\|$ for all states~$x'$ that can be generated by the perturbed system~\eqref{eq: pert sys}, so long as the small-gain condition~\eqref{eq:smallgaincondition} holds. The result also reveals an interplay between~$c$, a measure for the robustness of the baseline system, and the tunable parameter~$\rho$. A small~$c$, corresponding to a fragile baseline, makes the small-gain 
condition~\eqref{eq:smallgaincondition} harder to satisfy, and while increasing~$\rho$ relaxes this condition, it degrades the performance 
bound~\eqref{eq:small-gain}.

An immediate corollary of Prop.~\ref{prop: small-gain} is that if the baseline system $f^\star$ is $\delta$ISS, then ensuring robust stability of the deployed system $\hat f$ can be reduced to verifying a sector bound on the policy difference $\|(\hat \pi-\pi^\star)(\hat x)\|$. 
In \S\ref{ss:dmp-diss}, we show how this offers a practical means of enforcing the stability, robustness, and timing constraints~\eqref{eq:meta-timing}--\eqref{eq:meta-opt} of the DMP proposed in~\S\ref{s:qframe}.  For now we defer this discussion, and focus on bounding the {performance gap} of the deployed system $\hat f$.

\begin{remark}%
We note that if instead, $\|\Delta(x')\|\leq L\|x'\| + d$,\arxivver{\\}{ }for some $d>0$, then Prop.~\ref{prop: small-gain} can be suitably modified to show \emph{incremental practical stability} of $\|x-x'\|,$ i.e.,~\eqref{eq:small-gain} instead becomes $\|x'_t-x_t\|\leq \alpha_1^{-1}(e^{-c(1-\rho)t}\alpha_2(\|x_0\|)) + \gamma(\sup_{0\leq \tau \leq t}\|d_t\|)$ for~$\gamma$ some class~$\Kcal_\infty$ function.  Similarly, if $u_t\neq 0$ for all~$t$, then the error system inherits a (suitably degraded) ISS property of the nominal system.  Both follow from a standard comparison lemma with inputs argument.
\end{remark}

\subsection{Bounding the Performance Gap}%

Our next goal is to bound the performance gap
\begin{equation}\label{eq:delta_J}
    J(\hat\pi)-J(\pi^\star) = \arxivver{\textstyle\int_0^\infty}{\int_0^\infty} c(\hat x_t,\hat \pi(\hat x_t))-c(x^\star_t,\pi^\star(x^\star_t))\mr{d}t
\end{equation}
for a generic deployed policy $\hat \pi$ satisfying $\|(\hat \pi-\pi^\star)(\hat x)\|\leq \hat{L}\|\hat x\|$ for all $\hat x\in\R^\dx$ generated by the perturbed system~$\hat f$, and for a baseline system~$f^\star$ satisfying the assumptions of Prop.~\ref{prop: small-gain}, i.e.,~$\|\hat x_t-x^\star_t\|$ satisfies bound~\eqref{eq:small-gain}.  We further assume that there exists some function $\mu_\rho\in\Kcal_\infty$ such that 
\begin{equation}\label{eq:mu}
    \arxivver{\textstyle\int_0^\infty}{\int_0^\infty} \alpha_1^{-1}(e^{-{c(1-\rho)}t}\alpha_2(\|x_0\|))\mathrm{d}t\leq \mu_\rho(\|x_0\|).
\end{equation}
This assumption is true for most reasonable $\alpha_1$,~$\alpha_2$ functions, e.g., any $\alpha_i$ functions that grow at most polynomially in their arguments. Moreover, it is easy to show that %
the left-hand side of~\eqref{eq:mu} is monotonically increasing with respect to~$\rho$.

The following lemma shows, under suitable regularity assumptions on the stage-wise cost function $c(\cdot,\cdot)$ and the assumption that $\|(\hat \pi-\pi^\star)(x^\star)\|\leq L^\star\|x^\star\|$ for all~$x^\star$ generated by the baseline system~$f^\star$, that the {performance gap}~\eqref{eq:delta_J} scales at most linearly with $L^\star$.

\begin{lemma}\label{lem:delta_J}
    Assume that $\|(\hat \pi-\pi^\star)(\hat x)\|\leq \hat{L}\|\hat x\|$ for any~$\hat x$ generated by the perturbed system~\eqref{eq: pert sys}, and that $\|(\hat \pi-\pi^\star)(x^\star)\|\leq L^\star\|x^\star\|$ for any~$x^\star$ generated by the optimal system~$f^\star$.  Further, assume that the cost function satisfies
    \begin{align}\label{eq:cost_x}
    |c(\hat x,\hat \pi(\hat x))-c(x^\star,\hat \pi(x^\star))|&\leq L_\mr{x}\|\hat x-x^\star\|, \\
	\label{eq:cost_pi}
    	|c(x^\star,\hat \pi(x^\star))-c(x^\star,\pi^\star(x^\star))|&\leq L_\pi\|(\hat \pi-\pi^\star)(x^\star)\|,
	\end{align}
	for all $\hat x,x^\star\in\R^\dx$.  Then, if the assumptions of Prop.~\ref{prop: small-gain} are satisfied with $L^\star$, and if upper bound~\eqref{eq:mu} holds,
\begin{equation}\label{eq:delta_J_bound}
    J(\hat \pi)-J(\pi^\star) %
    \leq (L_\mr{x} + L_\pi L^\star)\mu_\rho(\|x_0^{\star}\|).
\end{equation}
\end{lemma}
Analogously to Prop.~\ref{prop: small-gain}, Lemma~\ref{lem:delta_J} shows that bounding the {performance gap}~\eqref{eq:delta_J} can be reduced to characterizing a design-dependent sector bound quantified by~$L^\star$.  In contrast to the small-gain condition~\eqref{eq:smallgaincondition}, this sector bound must only be valid on trajectories \emph{generated by the baseline system~$f^\star$.}

\subsection{The Design Meta-Problem for $\delta$ISS Systems}\label{ss:dmp-diss}

Towards operationalizing these conditions and bounds, we again leverage the sequential decomposition approach used to obtain~\eqref{eq:app-error}, but this time applied to $\|\hat\pi_N-\pi^\star\|$. %
Towards that end, consider the following policy error decomposition:
\arxivver{%
\begin{multline}\label{eq:policydecomposition}
\|(\hat\pi_N-\pi^\star)(x)\|\leq\|(\hat\pi_N - \hat\pi_\tau)(x)\| + \|(\hat\pi_\tau - \hat\pi_T)(x)\| \\+ \|(\hat\pi_T-\hat\pi_{\Mcal})(x)\| + \|(\hat\pi_{\Mcal}-\pi^\star)(x)\| \\
\leq L(N,\tau,T,\Mcal)\|x\|,%
\end{multline}
where $L(N,\tau,T,\Mcal):=( L_N + L_\tau + L_T + L_\Mcal)$.}{%
\begin{align}
\|(\hat\pi_N-\pi^\star)(x)\|&\leq\|(\hat\pi_N - \hat\pi_\tau)(x)\| + \|(\hat\pi_\tau - \hat\pi_T)(x)\| + \|(\hat\pi_T-\hat\pi_{\Mcal})(x)\| + \|(\hat\pi_{\Mcal}-\pi^\star)(x)\| \notag\\
&\leq \underbrace{( L_N + L_\tau + L_T + L_\Mcal)}_{=:L(N,\,\tau,\,T,\,\Mcal)}\|x\|.\label{eq:policydecomposition}
\end{align}}%
In the above, we associate a sector bound with gain~$L_\bullet$ to each of the terms in the decomposition, each of which captures the effect of different design decisions. For example, $L_\tau$ captures the error introduced due to temporal discretization. We re-emphasize that, based on the \emph{order} of the decomposition, the gains later in the decomposition are inherently dependent on the former design parameters. For example, for the introduced order in~\eqref{eq:policydecomposition}, $L_\tau$ will also depend on the choices for~$T$ and~$\Mcal$. 
Finally, we distinguish between $\hat{L}(N,\tau,T,\Mcal)$, which holds on states generated by the system evolving under~$\hat\pi_N$, and $L^\star(N,\tau,T,\Mcal)$, which holds on states generated by the baseline system. 

Substituting the policy decomposition~\eqref{eq:policydecomposition} 
into the sector bound argument of Prop.~\ref{prop: small-gain} 
and the performance bound of Lemma~\ref{lem:delta_J} yields a unified 
condition that certifies both stability and performance of the deployed 
policy $\hat\pi_N$, as formalized in the following central result.
\begin{theorem}\label{thm:dmp}
    Fix $\rho\in(0,1)$, and suppose that $f^\star$ is $\delta$ISS with 
    $\delta$ISS-Lyapunov function $V_\delta$ satisfying~\eqref{eq: lyap 
    comp fcns} and~\eqref{eq: diss}, characterized by $\alpha_1, 
    \alpha_2, \sigma \in \Kcal_\infty$ and constant $c > 0$. Further 
    suppose that the cost function satisfies~\eqref{eq:cost_x} 
    and~\eqref{eq:cost_pi}, and let $\mu_\rho \in \Kcal_\infty$ be such 
    that~\eqref{eq:mu} holds. Suppose additionally that the policy 
    decomposition~\eqref{eq:policydecomposition} holds with gains $L^\star$ and $\hat L$, where 
    $L^{\star}(N,\tau,T,\Mcal) := L^{\star}_N + L^{\star}_\tau + 
    L^{\star}_T + L^{\star}_\Mcal$ holds on trajectories 
    of the baseline system $f^\star$, and similarly, $\hat L_\bullet$ holds on trajectories 
    of the deployed system $\hat f$.
   \vspace{-1pt}
    Then, any solution $(N, \tau, T, \Mcal)$ to\arxivver{%
    \begin{subequations}\label{eq:meta-problem-general}%
    \begin{align}
         \minimize{N,\,\tau,\,T,\,\Mcal} \  & \left(L^\star_N + 
         L^\star_\tau + L^\star_T + 
         L^\star_\Mcal\right)\mu_\rho(\|x_0\|) \label{eq:meta-obj-diss}  \\[-3pt]
         \subjto \ & \hat L_N + \hat L_\tau + \hat L_T + \hat L_\Mcal \leq
         {\textstyle\frac{\sigma^{-1}(\rho c\alpha_1(z))}{2z}}, \, \forall z>0
         \label{eq:meta-small-gain-diss} \\[-2pt]
                            & \Tau(N,\tau,T,\Mcal)\leq \tau,
                            \label{eq:meta-compute-diss} 
    \end{align}
    \end{subequations}}{%
    \begin{subequations}\label{eq:meta-problem-general}%
    \begin{align}
         \minimize{N,\,\tau,\,T,\,\Mcal} \quad  & \left(L^\star_N + 
         L^\star_\tau + L^\star_T + 
         L^\star_\Mcal\right)\mu_\rho(\|x_0\|) \label{eq:meta-obj-diss}  \\
         \subjto \quad & \hat L_N + \hat L_\tau + \hat L_T + \hat L_\Mcal \leq
         {\frac{\sigma^{-1}(\rho c\alpha_1(z))}{2z}}, \ \text{for all } z>0
         \label{eq:meta-small-gain-diss} \\
                            & \Tau(N,\tau,T,\Mcal)\leq \tau,
                            \label{eq:meta-compute-diss} 
    \end{align}
    \end{subequations}
    }%
    yields a deployed policy $\hat\pi_N$ that: \emph{(i)} satisfies the 
    feasibility constraints~{\normalfont(\ref{eq:meta-timing})--(\ref{eq:meta-compute})} 
    of the DMP, and \emph{(ii)}~achieves a performance gap bounded by
    \begin{equation}\label{eq:perfgap_prop}
        J(\hat\pi_N) - J(\pi^\star) \lesssim %
         \left(L^\star_N + L^\star_\tau 
        + L^\star_T + L^\star_\Mcal\right)\mu_\rho(\|x_0\|),
    \end{equation}
    with $x_0$ %
    the initial condition of optimal control problem~\eqref{eq:base-problem}.
\end{theorem}

Theorem~\ref{thm:dmp} reduces the DMP to a tractable optimization over design-dependent sector bound gains $\hat L, L^\star$, with stability 
enforced via the small-gain constraint~\eqref{eq:meta-small-gain-diss} 
and computational feasibility via~\eqref{eq:meta-compute-diss}. While the 
resulting performance bound is explicit in
the initial condition $\|x_0\|$, it enters as a constant scaling factor, and we include it to highlight the tradeoff between the objective~\eqref{eq:meta-obj-diss} and the small-gain constraint~\eqref{eq:meta-small-gain-diss}, controlled by $\rho$. We also note that the framework is intentionally general and hence conservative.  For example, the 
small-gain condition of Prop.~\ref{prop: small-gain} is sufficient but not necessary. We expect tighter bounds can be derived for
specific problem classes by exploiting additional structure. 

A key remaining challenge is deriving expressions for the gains $L^\star,\hat L$ that make the role of the design parameters $(N,\tau,T,\Mcal)$ explicit. In \S\ref{s:lqrcasestudy}, we show this can be done for a 
receding horizon LQR problem.%

\section{Compute Constraint Receding Horizon LQR} %
\label{s:lqrcasestudy}

We apply our proposed framework to a receding horizon LQR problem. While the LQR problem admits an exact solution to the optimal control problem~\eqref{eq:base-problem} that can be implemented online with minimal computational overhead, it nevertheless provides a tractable and analytically transparent setting in which to validate the framework, and naturally connects to existing work on suboptimal MPC. The results of this section can be seen as extending this prior work to jointly optimize for model reduction, temporal discretization, and horizon selection in addition to finite algorithm iterations.  These results additionally lay the groundwork for analyzing more complex schemes, including linear MPC with constraints and sequential linearization-based methods such as iterative LQR.

\subsection{Baseline LQR Controller}\label{ss:baselinelqr}
The baseline LQR problem is:
\arxivver{%
\begin{subequations}\label{eq:lqr-base}
\begin{align}
     \minimize{\pi} \quad & {\textstyle\int_0^\infty} (x_t^\top Qx_t + u_t^\top Ru_t)\mathrm{d}t  \\[-2pt]
     \subjto \quad &  \dot x_t = Ax_t+Bu_t \label{eq:lqr-base:model}\\[-2pt]
     & u_t = \pi(x_t), \ x_0 = \xi, %
\end{align}
\end{subequations}}{%
\begin{subequations}\label{eq:lqr-base}
\begin{align}
     \minimize{\pi} \quad & {\int_0^\infty} (x_t^\top Qx_t + u_t^\top Ru_t)\mathrm{d}t  \\
     \subjto \quad &  \dot x_t = Ax_t+Bu_t \label{eq:lqr-base:model}\\
     & u_t = \pi(x_t), \\
     & x_0 = \xi,
\end{align}
\end{subequations}}%
where $A\in\R^{\dx \times \dx}$, $B\in\R^{\dx \times \du}$, and $Q,R\succ 0$ are positive definite matrices defining a quadratic stage cost~$c(x,u).$ We assume that $(A,B)$ is {controllable}, $(A,Q^{1/2})$ is {observable}, and that $A$ is invertible.\footnote{We assume $Q\succ 0$ and $A$ being invertible for clarity of exposition, but these assumptions can easily be relaxed.\label{fn:invertibleA}}

We use the optimal solution to~\eqref{eq:lqr-base} as the baseline controller, i.e.,~$\pi^\star(x) = K_\star x$, where
\(%
     K_\star := -R^{-1}B^\top P_\star,
\) %
and~$P_\star$ is the unique positive definite solution to the \emph{continuous-time algebraic Riccati equation}~(CARE)
\begin{equation}\label{eq:care}
    A^\top P_\star + P_\star A -P_\star BR^{-1}B^\top P_\star + Q = 0.
\end{equation}

It is well-known that for LTI systems, $\delta$ISS is equivalent to exponential stability~\cite{angeli2002lyapunov}. The following characterizes a function $\sigma$ and constant $c$ needed to define the small-gain condition~\eqref{eq:meta-small-gain-diss} for the optimal closed-loop LQR system. %

\begin{proposition}\label{prop:lqr-ediss}
    Denoting the optimal closed-loop matrix $A_{\textsc{lqr}}:=A+BK_\star$, the system $\dot x = A_{\textsc{lqr}}x + Bu$ is $\delta$ISS, and $V_\delta(x,x'):=(x-x')^\top  P_\star (x-x')$, where $P_\star$ is the solution to the CARE~\eqref{eq:care}, is a valid $\delta$ISS Lyapunov function with Lie derivative satisfying the exponential dissipative form~\eqref{eq: diss} with $c=\frac{\lmin(Q)}{\lmax(P_\star)}$ and $\sigma(z)=\lmax(R)z^2$.
\end{proposition}

\arxivver{The}{As is shown in Prop.~\ref{prop:boundingcost} (on page~\pageref{prop:boundingcost}), the}
cost function $c(x,K_\star x) = x^\top(Q + K_\star^\top RK_\star)x$ satisfies the regularity conditions required by Lemma~\ref{lem:delta_J}, so long as $x^\star_t,\hat x_t$ remain in a compact domain,  which is ensured by stability of both the baseline and deployed policies. %
Further, since $\alpha_1$ and $\alpha_2$ are quadratic, $\mu_\rho$ 
in~\eqref{eq:mu} exists and can be computed analytically as $\mu_\rho(z) := \frac{2\kappa(P_\star)^{1/2}}{c(1-\rho)}\,z,$ where \arxivver{$c$}{$c = \frac{\lmin(Q)}{\lmax(P_\star)}$} is as 
given in Prop.~\ref{prop:lqr-ediss}. Note again that $\mu_\rho$ is 
monotonically increasing in $\rho$, reflecting the 
robustness-performance tradeoff identified in \S\ref{ss:smallgain}.

\subsection{Design Approximations for Receding Horizon LQR}\label{ss:lqrdesignapprox}

Our strategy is to introduce these design approximations in the following order: model reduction, discretization, and finite horizon approximation (which subsumes solver iterations in this setting).  While related results exist in the literature, e.g., perturbation 
bounds for the algebraic Riccati equation~\cite{konstantinov1993perturbation} and convergence 
rates for the {discrete-time algebraic Riccati difference equation} (DARE)~\cite{hager1976convergence}, these must be adapted to yield explicit analytic expressions for how each 
approximation contributes to the sector bound gains of Theorem~\ref{thm:dmp}.\footnote{All derivations are deferred to  \arxivver{the extended version~\cite[Appendix~B]{extended}.}{Appendix~\ref{appendix:lqr-details}.}} In what follows, $\hat L$ (and corresponding constants~$\hat C_i$) denotes sector bound gains evaluated on 
trajectories of the deployed system $\dot{\hat x}_t = A\hat x_t + 
B\hat\pi_N(\hat x_{\tau\lfloor t/\tau\rfloor})$, and~$L^\star$ (and corresponding constants~$C^\star_i$) denotes 
gains evaluated on trajectories of the baseline system $\dot x^\star_t 
= A_{\textsc{lqr}}x^\star_t$. 

\begin{enumerate}[wide, leftmargin=0pt,label={\arxivver{\em}{\bf}\arabic*)},noitemsep,topsep=0pt]
    \item {\arxivver{\em}{\bf}Reduced Order Modeling:}
We consider a linear model and policy class composed of four models: $\Mcal:=\{M_0, M_1,M_2,M_3\}$, each with state-dimension~$n_{\mr{x},i}$. %
Model $M_0$ is the true system, and models $M_1,M_2,M_3$ are obtained using balanced truncation~\cite{antoulas2005approximation}.
We compute the optimal controller~$K_{\textsc{r}_i}$ for the reduced order model $(A_{\textsc{r}_i},B_{\textsc{r}_i})$ and $(Q_{\textsc{r}_i}, R)$ using~\eqref{eq:care}. We denote the negative spectral abscissa of the closed-loop for model $M_i$ by~$\alpha_i:=-\max_j \Re\{\lambda_j(A_{\textsc{r}_i}+B_{\textsc{r}_i}K_{\textsc{r}_i})\}$. %
For model $M_i$, 
\( \|(\hat\pi_{M_i}-\pi^\star)(x_t)\|\leq \|R^{-1}B(U^\top P_{\textsc{r}_i}U-P_\star)\|\|x_t\|,\) 
where~$U$ is the projection matrix obtained from the reduction method and $P_{\textsc{r}_i}$ is the solution to the CARE with the reduced model, and thus $\hat L_\Mcal(M_i)=L^\star_\Mcal(M_i)=\|R^{-1}B(U^\top P_{\textsc{r}_i}U-P_\star)\|$. %
    
\item {\arxivver{\em}{\bf}Discretization:} 
We bound the errors induced by a {zero-order hold} (ZOH) discretization\footnote{More sophisticated discretization may lead to different scalings with~$\tau$.} with sampling time~$\tau$ by a decomposition of the term $\|(\hat \pi_\tau-\hat \pi_{M_i})(x)\|$. We denote a ZOH signal by~$\xdt$ ($=x_{\tau\lfloor t/\tau\rfloor}$). To this end, let $\hat\pi_\textsc{zoh}(\tau,x_t)=K_{\textsc{r}_i} \xdt$ denote a ZOH approximation to~$\hat \pi_{M_i}$, and $\hat\pi_\tau(x_t)=K_\tau \xdt$ denote the optimal discrete-time LQR controller~$K_\tau$, obtained through the DARE, applied to the ZOH state $\xdt$ obtained via a first-order Euler discretization of the $M_i$ dynamics. Then, we have\arxivver{%
\vspace{-.75em}
\begin{multline*}
	\|(\hat\pi_\tau-\hat\pi_{M_i})(x_t)\|=\|(K_\tau-K_{\textsc{r}_i})\xdt + K_{\textsc{r}_i}(\xdt-x_t)\| \\ \qquad \leq \|K_\tau-K_{\textsc{r}_i}\|\|\xdt\| + \|K_{\textsc{r}_i}\|\|\xdt-x_t\|\\[-2em]
\end{multline*}}{%
\[\|(\hat\pi_\tau-\hat\pi_{M_i})(x_t)\|=\|(K_\tau-K_{\textsc{r}_i})\xdt + K_{\textsc{r}_i}(\xdt-x_t)\| \leq \|K_\tau-K_{\textsc{r}_i}\|\|\xdt\| + \|K_{\textsc{r}_i}\|\|\xdt-x_t\|\]
}%
The first term captures the mismatch between the continuous and discrete optimal gains, and the second captures the state quantization error due to ZOH sampling. Both scale as $O(\tau)$ for small enough $\tau$, reflecting the first-order accuracy of Euler discretization, and hence there exists system dependent constants~$C_1^\star$ such that $\|(\hat \pi_\tau-\hat \pi_{M_i})(x^\star_t)\|\leq C_1^\star\tau\|x^\star_t\|$ for all small enough~$\tau$, and similarly for $\hat C_1$ with respect to $\|\hat x_t\|$. %

\item {\arxivver{\em}{\bf}Finite Horizon Approximation:} 
For the receding horizon LQR controller, we set the discrete-time horizon to $N:=\floor{\tfrac{T}{\tau}}$.  The effects of finite horizon approximation and solver iterations can be combined, as computing the solution to the $N$-horizon LQR problem requires taking~$N$ recursions of the {discrete-time algebraic Riccati recursion} (DARR). Let $\hat\pi_T(x_t)=K_{T,\tau}\xdt$ denote the policy computed from a finite horizon approximation of the DARE solution with a horizon set to $N$ and terminal cost $c_T(x_t)=\xdt^\top \tau Q\xdt$. Then, standard convergence rates for the DARE over~$N$ iterations~\cite{hager1976convergence}, along with our discretization bounds outlined above, show that for small enough~$\tau$, there exist system dependent constants $C^\star_2,C^\star_3$ such that $\|(\hat\pi_T-\hat\pi_\tau)(x^\star_t)\|\leq (C^\star_2 + C^\star_3 \tau) e^{-\alpha_i T}\|x^\star_t\|$, and similarly for $\hat C_2, \hat C_3$ with respect to $\|\hat x_t\|$. The exponential convergence rate of the DARR is governed by $\alpha_i$, the least-stable closed-loop pole of $M_i$, highlighting that a more stable closed-loop requires a shorter horizon to achieve the same accuracy. {We deploy~$\hat\pi_T(x_t)=K_{T,\tau}\xdt$ in a receding horizon fashion, recomputing $K_{T,\tau}$ with the DARR at each sampling time.}
    
\item {\arxivver{\em}{\bf}Computation Model:} 
Given the state dimension~$n_{\mr{x},i}$ of~$M_i$, we consider the following computation model. At each discrete time step we solve the DARR over~$N(\leq T/\tau)$ steps, where each step incurs the $4n_{\mr{x},i}^3+4n_{\mr{x},i}^2\du+2n_{\mr{x},i}\du^2+\frac{1}{3}\du^4=:\Phi(M_i)$ flops required to compute the matrix multiplication, addition, and inversions defining the DARR. Assigning a machine-dependent constant of~$\tau_g$ seconds/flop, we define the timing function to be~$\Tau(\tau,T,M_i)= %
\tau_g \tfrac{T}{\tau} \Phi(M_i)$. %

\end{enumerate}

\subsection{Compute Constrained Receding Horizon LQR DMP}

Combining the parametric forms of the sector bounds derived in \S\ref{ss:lqrdesignapprox} via the 
policy decomposition~\eqref{eq:policydecomposition} yields:%
\arxivver{\vspace{-5pt}}{}
\arxivver{%
\begin{multline}\label{eq:pidecomposition_lqr}
	\|(\hat\pi_T-\pi^\star)(x_t)\|\leq \\ \big(C_1\tau + (C_2 + C_3\tau)e^{-\alpha_i T} + L_\Mcal(M_i)\big)\|x_t\|,
\end{multline}}{%
\begin{equation}\label{eq:pidecomposition_lqr}
	\|(\hat\pi_T-\pi^\star)(x_t)\|\leq  \left(C_1\tau + (C_2 + C_3\tau)e^{-\alpha_i T} + L_\Mcal(M_i)\right)\|x_t\|,
\end{equation}
}%
with system-dependent constants $C_i$, which can be identified from, e.g., data.
The parametric form~\eqref{eq:pidecomposition_lqr} holds for all $\tau\leq\tau_\star$, with $\tau_\star$ sufficiently small for both $(L^\star_\Mcal,C_i^\star)$ and $(\hat L_\Mcal,\hat C_i)$.
A key observation is that, for fixed $M_i$, the timing  constraint~\eqref{eq:meta-compute-diss} is always active at optimality, as it is never beneficial to use less than the available compute budget. We can therefore set $T = T_\star(\tau, M_i) := 
\frac{\tau^2}{\tau_g \Phi(M_i)}$, which saturates the timing 
constraint~\eqref{eq:meta-compute-diss}, reducing the DMP to a scalar 
optimization over $\tau \in (0, \tau_\star]$ for each $M_i$.
\begin{theorem}\label{thm:dmp-lqr}
	Consider the design problem~\eqref{eq:lqr-base} subject to the design approximations presented in~\S\ref{ss:lqrdesignapprox}. Fix a $\rho\in(0,1)$, a model $M_i$, $i\in\{0,1,2,3\},$ let $T_\star(\tau,M_i):=\frac{\tau^2}{\tau_g\Phi(M_i)}$, and $P_\star$ be the solution of~\eqref{eq:care}. Then, for small enough~$\tau_\star$, the $\delta$ISS DMP~\eqref{eq:meta-problem-general} with fixed model $M_i$ is given by:\arxivver{%

    \vspace{-1.5em}
    
	\begin{subequations}\label{eq:lqr-dmp-tau}
	\begin{align}
	    \minimize{0 < \tau \leq \tau_\star}  \  &  \tfrac{C^\star_1}{1-\rho}\tau + \tfrac{C_2^\star+C_3^\star\tau}{1-\rho}e^{-\alpha_i T_\star(\tau,M_i)}+\tfrac{L_\Mcal(M_i)}{1-\rho} \label{eq:lqr-dmp-tau-cost} \\[-3pt]
	     \subjto \ &  \hat C_1\tau + (\hat C_2 + \hat C_3\tau)e^{-\alpha_i T_\star(\tau,M_i)} \notag \\[-3pt] & \hspace{1.3cm} \leq \textstyle\frac{1}{2}\sqrt{\frac{\rho\lmin(Q)}{\cn(P_\star)\lmax(R)}} - L_\Mcal(M_i) \label{eq:lqr-dmp-tau-sgc}
	\end{align}
	\end{subequations}}{%
	\begin{subequations}\label{eq:lqr-dmp-tau}
	\begin{align}
	    \minimize{0 < \tau \leq \tau_\star}  \quad   &  \frac{C^\star_1}{1-\rho}\tau + \frac{C_2^\star+C_3^\star\tau}{1-\rho}e^{-\alpha_i T_\star(\tau,M_i)}+\frac{L_\Mcal(M_i)}{1-\rho} \label{eq:lqr-dmp-tau-cost} \\
	     \subjto \quad &  \hat C_1\tau + (\hat C_2 + \hat C_3\tau)e^{-\alpha_i T_\star(\tau,M_i)} \leq \textstyle\frac{1}{2}\sqrt{\frac{\rho\lmin(Q)}{\cn(P_\star)\lmax(R)}} - L_\Mcal(M_i) \label{eq:lqr-dmp-tau-sgc}
	\end{align}
	\end{subequations}}%
\end{theorem}
For each $M_i$, problem~\eqref{eq:lqr-dmp-tau} reduces to a 
scalar optimization over the compact interval $(0,\tau_\star]$, which 
can be solved by dense gridding or any scalar optimization routine. 
The overall DMP is then solved by comparing the optima across 
$M_0, M_1, M_2, M_3$, directly identifying the jointly optimal reduced order model, sampling rate, and finite horizon choice.

\subsection{Numerical Example}\label{ss:numerical}

We operationalize Theorem~\ref{thm:dmp-lqr} on a {large, randomly generated LTI system of the form~\eqref{eq:lqr-base:model} with $\dx = 97$ and $\du=1$.\footnote{All code available at \url{https://tinyurl.com/sxvjp8x4}.} The system is then normalized to have one unstable pole at $0.4$, and two slow stable poles at $-0.5$ and $-1.2$. The remaining stable poles are set to be very fast, with $\Re\{\lambda_i\}<-30$. For the control design problem,}  we choose $Q=R=I$, and set $\tau_g = 6.25\cdot10^{-8}$, corresponding to the standard 16MHz clock speed of an Arduino Leonardo. {The orders for the reduced order models $M_i$ are chosen as $(n_{\mr{x},1},n_{\mr{x},2},n_{\mr{x},3})=(6,3,1)$.}

Our analysis established that $L^\star$ 
and $\hat L$ take the functional form~\eqref{eq:pidecomposition_lqr}, 
with constants $C^\star_i$ and $\hat C_i$ to be identified.  \arxivver{While~\cite[App.~II]{extended}}{Appendix~\ref{appendix:lqr-details}} provides analytical bounds on these constants, they can be overly conservative. We therefore instead
estimate these constants from 10 randomly selected values of $\tau$ by 
computing $\|K_{\tau, T_\star} - K_\star\|$, with appropriate corrections 
for ZOH effects. %
The resulting design-dependent sector bound gains
$L^\star(\tau)$ and $\hat L(\tau)$ are then evaluated on a dense grid  $\tau\in[10^{-4},0.04]$ to solve~\eqref{eq:lqr-dmp-tau} and identify the optimal $\tau^\star$ (and hence $T_i^\star=T_\star(\tau^\star,M_i)$).  We repeat this for each model $M_{\{0,1,2,3\}}$, and compare the resulting designs to identify a jointly optimal design tuple $(\tau^\star, M^\star, T^\star)$. %

\begin{figure}
\centering
\subfloat[Reduced order model~$M_1$ with $n_{\mr{x},1}=6$]{%
	\arxivver{\includegraphics[scale=.95,trim={0 0.0mm 0 0.5mm},clip]{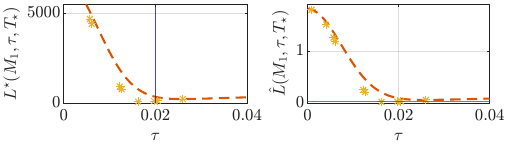}}{\includegraphics[width=0.95\linewidth]{figures/suboptM1}}
	\label{fig:suboptM1}}

\vspace{-.5pt}

\subfloat[Reduced order model~$M_2$ with $n_{\mr{x},2}=3$]{%
\arxivver{\includegraphics[scale=.95,trim={0 0.0mm 0 0.5mm},clip]{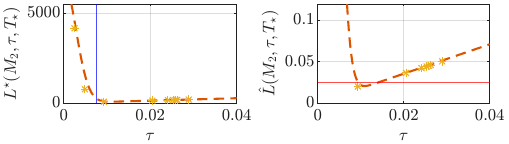}}{\includegraphics[width=0.95\linewidth]{figures/suboptM2}}
\label{fig:suboptM2}}

\vspace{0pt}

\subfloat[Reduced order model~$M_3$ with $n_{\mr{x},3}=1$]{%
\arxivver{\includegraphics[scale=.95,trim={0 0.0mm 0 0.5mm},clip]{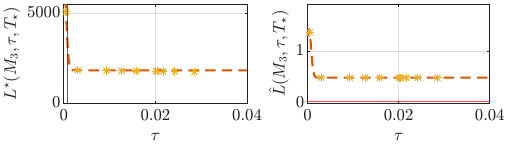}}{\includegraphics[width=0.95\linewidth]{figures/suboptM3}}
\label{fig:suboptM3}}
\caption{Computed values of the design-dependent sector bounds $L^\star$ and $\hat L$ for 10 randomly selected~$\tau$s (\textcolor{myellow}{$\ast$}) for the reduced order models~$M_1,M_2,M_3$, and the fitted function~\eqref{eq:pidecomposition_lqr} of the sector bound (\legendline{morange}\,\legendline{morange}). The right-hand side of the small-gain constraint~\mbox{(\legendline{red})} in~\eqref{eq:lqr-dmp-tau-sgc} indicates the feasible region for~$\tau$. The vertical blue line~\mbox{(\legendline{blue})} depicts the~$\tau$ that minimizes $\|A_\textsc{lqr} x_t-(Ax_t + BK_{T^\star\!\!,\tau}\xdt)\|$ for all~$x_t$.}\label{fig:subopt}

\vspace{-2em}

\end{figure}

We follow the procedure outlined above. {The interplay between compute resources and system dynamics lead to no feasible solution for the full order system $M_0$,  and required setting $\rho = 0.97$ to obtain feasible solutions to the DMP~\eqref{eq:lqr-dmp-tau}} for the remaining models $M_1,M_2,M_3$. The results are shown in Fig.~\ref{fig:subopt}. The left-hand plots show the performance gap sector bound gain $L^\star$, corresponding to the objective~\eqref{eq:lqr-dmp-tau-cost}, and the right-hand plots show the deployed policy sector bound gain $\hat L$, which must satisfy the small-gain condition~\eqref{eq:lqr-dmp-tau-sgc}. The right-hand side of~\eqref{eq:lqr-dmp-tau-sgc} is indicated by the red horizontal line. 
To provide a baseline, we use nonlinear optimization to compute the (locally) optimal sampling time that minimizes $\|A_\textsc{lqr} x_t-(Ax_t + BK_{T_\star,\tau}\xdt)\|$ for all~$x_t$, which is indicated by a blue vertical line.  One can view this as a DMP that directly minimizes the difference between the deployed and baseline closed-loop systems, rather than the difference between deployed and baseline policies. While appealing, this perspective heavily exploits the linear problem structure and does not obviously extend to the nonlinear setting addressed by Theorem~\ref{thm:dmp}. 

The minimizer of the fitted objective~\eqref{eq:lqr-dmp-tau-cost} closely coincides with the vertical blue line, suggesting that the empirical estimates of~\eqref{eq:lqr-dmp-tau-cost} and~\eqref{eq:lqr-dmp-tau-sgc} are accurate. This confirms that Theorem~\ref{thm:dmp-lqr} successfully identifies near-optimal design parameters, and validates the functional form of the sector bound~\eqref{eq:pidecomposition_lqr} derived \arxivver{in~\cite[App.~II]{extended}.}{in Appendix~\ref{appendix:lqr-details}.} Moreover, the plots in Fig.~\ref{fig:subopt} also demonstrate the tradeoffs across different axes, such as model order and compute resources. The smallest $(M_i,\tau)$-dependent sector bound gain (with respect to the fitted sector bound) is achieved for the optimal design tuple $(M^\star,\tau^\star)=(M_2,0.0178)$, with $T^\star=12.4$ seconds. %

We also explore the influence of $\rho$ on the solutions to~\eqref{eq:lqr-dmp-tau}. Fig.~\ref{fig:suboptrho} shows the cost function~\eqref{eq:lqr-dmp-tau-cost} (left plot) and the small-gain constraint~\eqref{eq:lqr-dmp-tau-sgc} (right plot, where green indicates satisfaction of the inequality~\eqref{eq:lqr-dmp-tau-sgc}) over a tight grid of $\tau\times\rho$. As~$\rho$ approaches 1, it becomes easier to satisfy the small-gain condition, indicated by a wider range of feasible~$\tau$s. However, the left plot shows that for larger~$\rho$, the performance gap grows exponentially, highlighting the robustness-performance tradeoff inherent to our framework. \arxivver{The extended version~\cite{extended} includes additional simulation results and studies the effect of increased compute.}{\par In Appendix~\ref{appendix:example}, we provide more background and analysis of this numerical study. In particular, we \emph{i})~elaborate on how we computed~$\|K_{\tau, T_\star} - K_\star\|$, i.e., the data points~(\textcolor{myellow}{$\ast$}) in Fig.~\ref{fig:subopt}, \emph{ii})~provide and compare numerical simulations of the deployed policies corresponding to the minima in Fig.~\ref{fig:subopt}, \emph{iii})~improve the functional form~\eqref{eq:pidecomposition_lqr} by including higher-order terms, and \emph{iv})~study the effect of increased computational resources on the feasibility and cost of the DMP.}

\begin{figure}
    \centering
    \arxivver{\includegraphics[scale=.9,trim={0 4.9mm 0 0.2mm},clip]{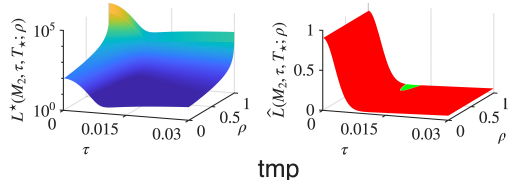}

    \vspace{-3pt}
    }{
    \includegraphics[width=1\linewidth,trim={0 4.9mm 0 0mm},clip]{figures/suboptrho}}
    \caption{Computed values of the design-dependent sector bound over a grid of $\tau\times\rho$ for the optimal and deployed systems for system~$M_2$. The region for which the small-gain condition~\eqref{eq:lqr-dmp-tau-sgc} is satisfied is indicated with green.}
    \label{fig:suboptrho}
    
    \vspace{-1.5em}
\end{figure}

\section{Conclusions and Outlook}\label{s:conclusion}\vspace{-0.2em}

We presented a framework for systematically navigating approximation choices in controller design under explicit computational constraints,  formalized as a Design Meta-Problem (DMP) that jointly optimizes controller performance subject to stability, real-time compute, and timing constraints. When the baseline system is $\delta$ISS, bounding closed-loop performance degradation reduces to verifying a design-dependent sector bound on the policy difference, with stability enforced via a small-gain condition. In the receding horizon LQR setting, this yields a scalar optimization over $\tau$ for each candidate model $M_i$, from which a jointly optimal $(\tau^\star, M^\star)$ is identified by direct comparison.

Looking ahead, we aim to extend the framework in three directions. 
First, for the LQR example, the sector bound gains $\hat L, L^\star$ were estimated by 
bounding errors in the feedback gains for various choices of $\tau$.  While this is tractable for linear feedback policies, it may not be in general, and developing general 
data-driven or empirical methods for fitting these gains, e.g., via sampling-based estimation of $\|(\hat{\pi} - 
\pi^\star)(x)\|$, is an important direction for broader applicability. 
Second, the DMP naturally accommodates additional 
approximations and constraints, such as memory constraints for explicit MPC or sector bound terms capturing model uncertainty from 
data-driven identification. 
Third, and more ambitiously, 
Theorem~\ref{thm:dmp} provides a principled building block for the 
systematic design of layered control 
architectures~\cite{matni2024quantitative}: when a single-layer DMP 
is infeasible, this motivates decomposing the problem into layers, 
each with its own compute resources and DMP, linked via inter-layer 
constraints.

\bibliographystyle{IEEEtran}
\bibliography{refs_cdc26tradeoff}

\arxivver{%
\clearpage
}{%
\appendices
\section{Proofs}\label{appendix:proofs}
\begin{proof}[Proof of Proposition~\ref{prop: small-gain}]
   We are interested in analyzing $x-x'$.
    By the assumption that $\dot x = f(x,u)$ is $\delta$ISS with $\delta$ISS-Lyapunov function~$V_\delta$ satisfying~\eqref{eq: lyap comp fcns} and inequality~\eqref{eq: diss}, we have that $\dot V_\delta$ satisfies
    \begin{align}
        \dot V_\delta(x,x') & \leq -cV_\delta(x,x') + \sigma(\|\Delta(x')\|) \notag \\
      & \leq -c V_\delta(x,x') + \sigma(2L\|x-x'\|) + \sigma(2L\|x\|)  \notag \\
      & \leq -c V_\delta(x,x') + \rho c \big(V_\delta(x,x')+V_\delta(x,0)  \big)\notag \\ 
      & \leq -c(1-\rho)V_\delta(x,x') + \rho c e^{-ct}V_\delta(x_0,0),\label{eq:VdotleqV+V0}
    \end{align}
    where the second inequality follows from $\|\Delta(x')\|\leq L\|x'\| \leq L\|x-x'\| + L\|x\|$ and $\sigma(a+b)\leq \sigma(2a) + \sigma(2b)$, 
    the third and fourth from the small-gain assumption $\sigma(2Lz)\leq \rho c \alpha_1(z)$ and~\eqref{eq: lyap comp fcns}, respectively. We get to~\eqref{eq:VdotleqV+V0} using a standard comparison lemma argument~\cite[Lem.~3.4]{khalil2002nonlinear} on~\eqref{eq: lyap comp fcns} for $x'=0$. %
    Reapplication of the comparison lemma on~\eqref{eq:VdotleqV+V0} gives
    \begin{align*}
        V_\delta(x_t,x_t') & \leq e^{-c(1-\rho)t}V_\delta(x_0,x_0')  + \rho c V_\delta(x_0,0)\left| {\int_0^t} e^{-c(1-\rho)(t-s)}e^{-cs}\mathrm{d}s \right| \\ 
        & =\rho c e^{-c(1-\rho)t}V_\delta(x_0,0) {\int_0^t} e^{-\rho c s}\mathrm{d}s  \\
        & =  e^{-c(1-\rho)t}\left( 1- e^{-\rho c t} \right)V_\delta(x_0,0) \\
        & \leq e^{-c(1-\rho)t}V_\delta(x_0,0),
    \end{align*}
    where we used that $V_\delta(x_0,x'_0)=0$.  Finally, we conclude:
    \[ \|x_t-x_t'\| \leq \alpha_1^{-1}\big(V_\delta(x_t,x_t')\big) %
    \leq \alpha_1^{-1}\big(e^{-c(1-\rho)t}\alpha_2(\|x_0\|)\big), \]
    which is the desired result.
\end{proof}
\begin{proof}[Proof of Lemma~\ref{lem:delta_J}]
    The following is immediate from~\eqref{eq:cost_x} and~\eqref{eq:cost_pi}, and Prop.~\ref{prop: small-gain}:
    \begin{align*}
        |c(\hat{x}_t,\hat{\pi}(\hat{x}_t))-c(x^\star_t,\pi^\star(x^\star_t))|
        &\leq |c(\hat{x}_t,\hat{\pi}(\hat{x}_t))-c(x^\star_t,\hat{\pi}(x^\star_t))| 
        + |c(x^\star_t,\hat{\pi}(x^\star_t))-c(x^\star_t,\pi^\star(x^\star_t))|\\ 
        &\leq L_\mr{x}\|\hat{x}_t-x^\star_t\| + L_\pi\|(\hat{\pi}-\pi^\star)(x^\star_t)\| 
        \leq L_\mr{x}\|\hat{x}_t-x^\star_t\| + L_\pi L^\star\|x^\star_t\| \\ %
        &\leq L_\mr{x}\alpha_1^{-1}\big(e^{-c(1-\rho)t}\alpha_2(\|x^\star_0\|)\big) + L_\pi L^\star \alpha_1^{-1}\big(e^{-ct}\alpha_2(\|x^\star_0\|)\big) \\
        &\leq {(L_\mr{x} + L_\pi L^\star)\alpha_1^{-1}\big(e^{-c(1-\rho)t}\alpha_2(\|x^\star_0\|)\big),}
    \end{align*}
    {where we used the fact that $\|x^\star_t\|\leq \alpha_1^{-1}(V_\delta(0,x^\star_t))\leq %
    \alpha_1^{-1}(e^{-ct}\alpha_2(\|x^\star_0\|))$.}
    Integrating and leveraging bound~\eqref{eq:mu} then yields the desired bound~\eqref{eq:delta_J_bound}.
\end{proof}

\begin{proof}[Proof of Theorem~\ref{thm:dmp}]
	Consider the decomposition~\eqref{eq:policydecomposition}.  Application of Prop.~\ref{prop: small-gain} and Lemma~\ref{lem:delta_J} allow us to conclude that for a fixed $\rho\in(0,1)$, 
	\begin{equation}\label{eq:bound-on-cost}
		J(\hat\pi_N)-J(\pi^\star)\leq(L_\mr{x} + L_\pi L^\star(N,\tau,T,\Mcal))\mu_\rho(\|x^\star_0\|)\propto L^\star(N,\tau,T,\Mcal)\mu_\rho(x_0),
	\end{equation}
	with $x_0$ as in~\eqref{eq:base-problem}, as long as $\sigma(2\hat{L}(N,\tau,T,\Mcal)z)\leq \rho c\alpha_1(z)$ for all $z\geq 0$. The bound~\eqref{eq:bound-on-cost} on~\eqref{eq:app-error} directly yields~\eqref{eq:meta-obj-diss}, while rewriting the small-gain condition for~$z\neq0$ yields the term in~\eqref{eq:meta-small-gain-diss}. This inequality constraint subsumes the stability and timing constraints~\eqref{eq:meta-timing}--\eqref{eq:meta-opt}, as it restricts~$(N,\tau,T,\Mcal)$ to values that ensure the computed policy~$\hat{\pi}$ is stabilizing. Adding~\eqref{eq:meta-compute} gives~\eqref{eq:meta-problem-general}.
\end{proof}

\begin{proof}[Proof of Proposition~\ref{prop:lqr-ediss}]
    Let $P_\star$ be the solution to~\eqref{eq:care}, and consider the candidate $\delta$ISS Lyapunov function $V(x,x')= e^\top P_\star e$, where we write $e = x-x'$.  Note that $\dot e = A_{\textsc{lqr}}e + Bd$, with $d=u-u'$. 
    It is immediate that
    \[ \lmin(P_\star)\|e\|^2 \leq V(x,x')\leq \lmax(P_\star )\|e\|^2,\]
    and we compute
    \begin{align*}
        \dot V & = \dot e^\top P_\star e + e^\top  P_\star \dot e \\
        &= e^\top (A_{\textsc{lqr}}^\top P_\star + P_\star A_{\textsc{lqr}})e + 2d^\top B^\top P_\star e\\
        &= e^\top (-Q-K_\star^\top RK_\star)e - 2d^\top RK_\star e \\
        & \leq -\|e\|_Q^2-\|K_\star e\|_R^2+2\|d\|_R\|K_\star e\|_R\\
        &\leq-\|e\|_Q^2-\|K_\star e\|_R^2 + \|K_\star e\|_R^2+\|d\|_R^2\\
        &\leq -\lmin(Q)\|e\|^2+\|d\|_R^2 \leq -\tfrac{\lmin(Q)}{\lmax(P_\star )}V+\|d\|_R^2.
    \end{align*}
In the above, the third equality follows from $P_\star $ being the solution to the CARE~\eqref{eq:care} and that $P_\star BR^{-1}B^\top P_\star =K_\star^\top  RK_\star$, the first inequality follows by Cauchy-Schwarz, the second by Young's inequality, and the final from $\|e\|_Q^2\geq\lmin(Q)\|e\|^2$ and $V(x,x')\leq\lmax(P_\star )\|e\|^2.$

Now, set $c=\frac{\lmin(Q)}{\lmax(P_\star )}$.  Then, via a standard application of the comparison lemma, we conclude that 
\[
V(x_t,x'_t)\leq e^{-ct}V(x_0,x'_0)+\tfrac{1}{c}\, {\textstyle\sup_{0\leq \tau \leq t}}\|d_\tau\|^2_R,
\]
and therefore, %
\[
\|e_t\|^2 \leq \cn(P_\star )\Big(e^{-ct}\|e_0\|^2 + \tfrac{\lmax(R)}{\lmin(Q)}{\textstyle\sup_{0\leq \tau \leq t}}\|d_\tau\|^2\Big).
\]
Taking square roots of both sides, and noting that $\sqrt{a+b}\leq\sqrt{a}+\sqrt{b}$ for $a,b>0$, we obtain the desired result.
\end{proof}

\begin{proof}[Proof of Theorem~\ref{thm:dmp-lqr}]
	The bound~\eqref{eq:pidecomposition_lqr} provides us with upper bounds for {$\hat{L}_\bullet^{\vphantom{\star}},L_\bullet^\star$} in~\eqref{eq:meta-problem-general} for all $0<\tau\leq\tau_\star$. With the given compute model, saturating the wall-clock inequality constraint~\eqref{eq:meta-compute-diss}, i.e.,  $\tfrac{1}{\tau}\tau_g T\Phi(M_i)\leq\tau$, is optimal for a fixed $M_i$, and hence we can set $T := T_\star(\tau,M_i)$. By Prop.~\ref{prop:lqr-ediss},~$\alpha_1$, $\alpha_2$, $c$, and~$\sigma$ are defined. Hence, the small-gain condition~\eqref{eq:meta-small-gain-diss} can be explicitly computed as
	\[ \text{right-hand side of~\eqref{eq:pidecomposition_lqr}} \leq \textstyle\frac{1}{2}\sqrt{\frac{\rho\lmin(Q)}{\cn(P_\star)\lmax(R)}}\|x_t\|. \]
	Dividing by $\|x_t\|$ and subtracting~$L_\Mcal(M_i)$  from both sides gives~\eqref{eq:lqr-dmp-tau-sgc}. Finally, we solve the left-hand side of~\eqref{eq:mu}, which gives~$\mu_\rho(z):=\frac{2\cn(P_\star)^{1/2}}{c(1-\rho)}z$. Absorbing~$\frac{2\cn(P_\star)^{1/2}}{c}$ in~$C^\star_i$ gives~\eqref{eq:lqr-dmp-tau-cost}.
\end{proof}

\section{Deriving the Design-Dependent Sector Bounds for the Receding Horizon LQR DMP}
\label{appendix:lqr-details}
This appendix discusses the detailed derivations for the receding horizon LQR design problem under compute constraints, as is presented in \S\ref{s:lqrcasestudy}. We follow the same order of discussion as in \S\ref{ss:lqrdesignapprox}, i.e., we first consider the derivation for the sector bound gain $L_\Mcal$ (\S\ref{appendix:modred}), followed by $L_\tau$ (\S\ref{appendix:discretization}), and $L_T$ (\S\ref{appendix:fhorizon}). Then in \S\ref{ass:costcond}, we show that the conditions on the cost in Lemma~\ref{lem:delta_J} hold. We conclude the appendix in~\S\ref{appendix:example} by providing additional details and extra plots for the numerical example presented in~\S\ref{ss:numerical}.

We first restate the assumptions associated with the baseline system and the LQR design problem at hand:
\begin{enumerate}
    \item The quadratic stage cost~$c(x_t,u_t)=x_t^\top Q x_t + u_t^\top R u_t$ is characterized by positive definite matrices $Q,R\succ 0$,\label{item:posdef}
    \item The true system~\eqref{eq:lqr-base:model} defined by $(A,B)$ is controllable, \label{item:contr}
    \item The $A$ matrix of the true system~\eqref{eq:lqr-base:model} is invertible. \label{item:inv}
\end{enumerate}
While Item~\ref{item:contr} is needed for unique, stabilizing solutions of~\eqref{eq:care}, we assume Item~\ref{item:posdef} for ease of notation. Under suitable regularity assumptions, we can relax this assumption and take the standard $Q\succeq 0$.  
As will be seen later in Prop.~\ref{prop:boundKtau-xi}, assuming Item~\ref{item:inv} significantly clarifies the exposition of the results. This assumption can also easily be lifted through the use of the Jordan form of $A$.

\subsection{Reduced Order Modeling}\label{appendix:modred}

The design-dependent sector bound corresponding to reduced order modeling is characterized by:
\[ \|(\hat \pi_{M_i}-\pi^\star)(x_t)\|\leq L_{\Mcal}(M_i)\|x_t\|  \]
for all $t\geq0$. Many techniques exist for obtaining reduced order models $M_i$, see~\citet{antoulas2005approximation}, most of which can be easily applied in our setting. We also want to highlight that there are specific methods available that consider state reduction with respect to the ARE~\citep{alla2019order}. Most balanced reduction methods (including~\citet{alla2019order}) can brought under the umbrella term of finding matrices $U,V\in\R^{\dx\times n_{\mr{x},i}}$ that satisfy a Petrov-Galerkin type condition,\footnote{In general, this conditions entails that $\mr{image}(V)$ must be $A$-invariant and $\mr{image}(B)\subseteq \mr{image}(V)$, see also~\citet[Chap.~7]{antoulas2005approximation}.} which project the state of the original system~$x_t$ to a state~$\tilde{x}_t$ with a lower dimension~$n_{\mr{x},i}$, where the state-space matrices of the reduced system are given by:
\[ A_\textsc{r} := U^\top A V, \quad B_\textsc{r} := U^\top B, \quad C_\textsc{r} := C V. \]
We assume that $U^\top V = I_{n_{\mr{x},i}}$, which is a required condition for most reduction methods for real systems, such as balanced truncation. Under this assumption, we have $x_t=V\tilde{x}_t$, with $\tilde{x}_t\in\R^{n_{\mr{x},i}}$, and $\tilde{x}_t = U^\top x_t$. In this setting, the reduced order LQR problem is
\begin{subequations}\label{eq:modredproblem}
\begin{align}
    \minimize{\hat\pi_{M_i}} \quad & \int_{0}^\infty \left(\tilde{x}_t^\top Q_\textsc{r}\tilde{x}_t + u^\top_t R u_t\right)\mr{d}t\\
         \subjto \quad & \tilde{x}_t = A_\textsc{r}\tilde{x}_t + B_\textsc{r}u_t, \label{eq:modredproblem:sys} \\ 
         & \tilde{x}_0 = \tilde{\xi} :=U^\top x_0 \\
         & u_t = \hat\pi_{M_i}(x_t) := K_{\textsc{r}_i} U^\top x_t
\end{align}
\end{subequations}
where $A_\textsc{r},B_\textsc{r}$ are defined as above, and $Q_\textsc{r}:= V^\top Q V$.  The resulting optimal control policy is given by $u_t=K_{\textsc{r}_i}\tilde{x}_t$, where $K_{\textsc{r}_i}=-R^{-1}B_\textsc{r}^\top P_{\textsc{r}_i}$, and $P_{\textsc{r}_i}$ is the solution to the CARE defined by $(A_\textsc{r},B_\textsc{r},Q_\textsc{r},R)$.
The following straight-forward result provides a way to compute~$L_{\Mcal}(M_i)$ for a given $U, V$:

\begin{proposition}
    For Petrov-Galerkin type matrices $U, V\in\R^{\dx\times n_{\mr{x},i}}$, let $P_{\textsc{r}_i}$ be the solution to the CARE that solves~\eqref{eq:modredproblem}, then
    \[ \|(\hat \pi_{M_i}-\pi^\star)(x_t)\|\leq \|R^{-1}B^\top(U P_{\textsc{r}_i} U^\top-P_\star)\|\|x_t\|, \]
    i.e., $L_{\Mcal}(M_i) := \|R^{-1}B^\top(U P_{\textsc{r}_i} U^\top-P_\star)\|$.
\end{proposition}
\begin{proof}
    The optimal control policy that solves~\eqref{eq:modredproblem} is given by $u_t=K_{\textsc{r}_i}\tilde{x}_t$, i.e.,
    \[ u_t=K_{\textsc{r}_i}\tilde{x}_t = -R^{-1}B_\textsc{r}^\top P_{\textsc{r}_i} \tilde{x}_t = -R^{-1}B^\top U P_{\textsc{r}_i} \tilde{x}_t = -R^{-1}B^\top U P_{\textsc{r}_i} U^\top x_t.\]
    Hence,
    \begin{multline*}
        \|(\hat\pi_{M_i}-\pi^\star)(x_t)\| = \|(-R^{-1}B^\top U P_{\textsc{r}_i} U^\top - K_\star)x_t \| = \|-R^{-1}B^\top (U P_{\textsc{r}_i} U^\top-P_\star)x_t\| \\\leq \|R^{-1}B^\top (U P_{\textsc{r}_i} U^\top-P_\star)\|\|x_t\|,
    \end{multline*}
    showing the claim.
\end{proof}
\begin{remark}
    As we mentioned before, there are many different ways to obtain the matrices~$U, V$. One of the simplest examples is to consider reduced order models that only consider the unstable modes of the original system. Under the assumption that~$A$ is diagonizable,  consider the eigenvalue decomposition of $A$:
    \[ A\Lambda = \Lambda\begin{bsmallmatrix} \lambda_1 & & \\ & \ddots & \\ & & \lambda_{\dx} \end{bsmallmatrix},\]
    and let $\Re\{\lambda_j\}\geq0$ for $j = 1, \dots, n_{\mr{x},i}$. Then, the reduced order model that only contains the unstable modes can be obtained with $U:= \Lambda^{-1}\begin{bmatrix} I_{n_{\mr{x},i}} & 0 \end{bmatrix}$ and $V := \Lambda\begin{bmatrix} I_{n_{\mr{x},i}} & 0 \end{bmatrix}$. In the numerical example in \S\ref{ss:numerical}, we have used a standard implementation of the balanced reduction method in Matlab to obtain $U,V$.\footnote{See also \url{https://tinyurl.com/sxvjp8x4}.}
\end{remark}

\newcommand{\Krom}{K_\textsc{r}}
\newcommand{\Prom}{P_\textsc{r}}
\subsection{Discretization}\label{appendix:discretization}
We bound the difference in policy due to ZOH discretization by characterizing the design-dependent sector bound:
\begin{equation}\label{eq:appdt:bound}
    \|(\hat\pi_{\tau}-\hat\pi_\Mcal)(x_t)\|\leq L_{\tau}\|x_t\| 
\end{equation}
for all $t\geq0$ and all reduced order models in $M_i\in\Mcal$. To avoid notational clutter, we denote the optimal controller $K_{\textsc{r}_i}$ for reduced order model $M_i$ by $\Krom$. Similarly, we denote $P_{\textsc{r}_i}$ by $\Prom$. We briefly revisit the approach discussed in \S\ref{ss:lqrdesignapprox}-2: we bound the errors induced by a ZOH discretization with sampling time~$\tau$ by a decomposition of the term $\|(\hat \pi_\tau-\hat \pi_{M_i})(x)\|$. We introduced $\hat\pi_\textsc{zoh}(\tau,x_t)=\Krom \xdt$ as the ZOH approximation to~$\hat \pi_{M_i}$ (recall the notation~$\xdt:=x_{\tau\floor{t/\tau}}$), and $\hat\pi_\tau(x_t)=K_\tau \xdt$ as the optimal discrete-time LQR controller~$K_\tau$, obtained through the DARE, applied to the ZOH state $\xdt$ obtained via a first-order Euler discretization of the $M_i$ dynamics. This gave for a given $M_i$:
\[
	\|(\hat\pi_\tau-\hat\pi_{M_i})(x_t)\|\leq\|\hat\pi_\tau(x_t)-\hat\pi_\textsc{zoh}(\tau,x_t)\|+\|\hat\pi_\textsc{zoh}(\tau,x_t)-\hat\pi_{M_i}(x_t)\|\leq \|K_\tau-\Krom\|\|\xdt\| + \|\Krom\|\|\xdt-x_t\|.
\]
Our goal is to identify upper bounds on $\|K_\tau-\Krom\|$ that scale with $\tau$, as well as upper bounds on $\|x_t-\xdt\|$ and $\|\xdt\|$ as a function of $\tau$ and $\|x_t\|$. We show that this is possible under suitable regularity assumptions and for small enough $\tau$. Most of the individual results in this section are coming from perturbation analysis techniques~\cite{konstantinov1993perturbation}, but have never been collectively used to achieve bounds of the form~\eqref{eq:appdt:bound}. We start by identifying upper bounds on $\|K_\tau-\Krom\|$ as a function of~$\tau$.

\subsubsection{Bounding {\normalfont$\|K_\tau-\Krom\|$}}

We first bound $\|K_\tau-\Krom\|$ as a function of $\|P_\tau - \Prom\|$ and $\tau$, i.e.,
\[ \|K_\tau-\Krom\|\leq O(\|P_\tau - \Prom\|) + O(\tau), \]
assuming both are sufficiently small. We then provide a bound on $\|P_\tau-\Prom\|$ as a function of~$\tau$. The following results that show this are standard and can be found in various forms in the literature~\cite{kenney1990sensitivity,konstantinov1993perturbation, bonnans2017error}. We included self-contained proofs of the results for convenience of the reader.

Define $K_\tau$ as the solution to the discrete LQR problem defined over first order Euler discretization of the dynamics and cost objectives in~\eqref{eq:modredproblem:sys}. We drop the $\textsc{r}$-subscript in the system and cost function matrices to reduce notational clutter. The first order Euler discretizations of the reduced order model $M_i$ and corresponding cost objective are parametrized by
\begin{equation}
    A_\tau = I + \tau A, \ B_\tau=\tau B, \ Q_\tau = \tau Q, \ R_\tau = \tau R.
\end{equation}
This means that $K_\tau=-(R_\tau + B_\tau^\top  P_\tau B_\tau)^{-1}B_\tau^\top P_\tau A_\tau$, where $P_\tau$ is the solution to the corresponding DARE.

First, we introduce the following mathematical tools that follow directly from the Von Neumann series:
\begin{lemma}\label{lem:neumannstuff}
    Let $A\in\mathbb{R}^{n\times n}$, with $\|A\|\leq\bar{a}<1$. Then, the following holds:
    \begin{enumerate}[label={(\roman*)}]
        \item $(I+A)^{-1}$ exists, \label{lem:item:1}
        \item $(I+A)^{-1}$ admits the expansion $I-A(I+A)^{-1}$, \label{lem:item:2}
        \item $\|(I+A)^{-1}\|\leq O(1)$. \label{lem:item:3}
    \end{enumerate}
\end{lemma}
\begin{proof}
    Item~\ref{lem:item:1} follows directly from the Von Neumann series argument. Item~\ref{lem:item:2} can be shown as follows:
    \[ (I+A)^{-1} = (I + A - A)(I+A)^{-1} = (I+A)(I+A)^{-1} - A(I+A)^{-1} = I-A(I+A)^{-1}. \]
    Finally for Item~\ref{lem:item:3}, we have that for $\|A\|\leq\bar{a}<1$, $\|(I+A)^{-1}\|\leq \frac{1}{1-\|A\|} \leq \frac{1}{1-\bar{a}}\leq O(1)$.
\end{proof}

With this result, we can derive a bound on $\|K_\tau-\Krom\|$ as a function of $\|P_\tau - \Prom\|$ and $\tau$:

\begin{proposition}\label{prop:k-bounds}
    Define $\Phi_\tau:= R^{-1/2}B^\top P_\tau BR^{-1/2}$. If, for the positive constants $\bar\tau$ and $\bar P$, we have that $\|\tau\Phi_\tau\|\leq\frac{1}{2}$ for all~$\tau\leq\bar\tau$ and for all $\|P_\tau\|\leq \bar P$, then
    it holds that
\begin{equation}\label{eq:K-bound}
\|K_\tau-\Krom\|\leq O(\|P_\tau-\Prom\|) + O(\tau).
\end{equation}
\end{proposition}
\begin{proof}
    We first observe that
    \begin{equation*}
        R_\tau + B_\tau^\top P_\tau B_\tau = \tau(R+\tau B^\top P_\tau B) = \tau R^{1/2}(I + \tau R^{-1/2}B^\top P_\tau BR^{-1/2})R^{1/2}=:\tau R^{1/2}(I + \tau \Phi_\tau)R^{1/2}.
    \end{equation*}
    For small enough~$\bar\tau$, $\|\tau\Phi_\tau\|\leq\frac{1}{2}<1$, Lemma~\ref{lem:neumannstuff} provides that $(I + \tau \Phi_\tau)^{-1}$ exists. This yields
    \begin{align*}
        K_\tau & = -\tfrac{1}{\tau}R^{-1/2}(I+\tau\Phi_\tau)^{-1}R^{-1/2}B_\tau^\top P_\tau A_\tau \\
        & = -R^{-1/2}(I+\tau\Phi_\tau)^{-1}R^{-1/2}B^\top P_\tau - \tau R^{-1/2}(I+\tau\Phi_\tau)^{-1}R^{-1/2}B^\top P_\tau A. 
    \end{align*}
    Substituting the expansion in Item~\ref{lem:item:2} of Lemma~\ref{lem:neumannstuff} and subtracting $\Krom = -R^{-1}B^\top \Prom$ gives:
    \begin{multline*}
        K_\tau-\Krom = -R^{-1}B^\top P_\tau + \\R^{-1/2} \tau\Phi_\tau (I+\tau\Phi_\tau)^{-1}R^{-1/2}B^\top P_\tau  + R^{-1}B^\top \Prom - \tau R^{-1/2} (I+\tau\Phi_\tau)^{-1}R^{-1/2}B^\top P_\tau A.
    \end{multline*}
    Taking the norm of $K_\tau-\Krom$ and repeatedly applying the triangle inequality and sub-multiplicative property, together with realizing that~$\|\Phi_\tau\|\leq\lmin(R)\|B\|^2\bar{P}\leq O(1)$ yields:
    \[ 
        \|K_\tau-\Krom\| \leq \underbrace{\|R^{-1}B^\top\|}_{\leq\lmin(R)\|B\|}\Big(\|P_\tau -\Prom\| + \tau \big(\underbrace{\|P_\tau\|}_{\leq\bar P}\underbrace{\|\Phi_\tau\|}_{\leq O(1)}\underbrace{\|(I+\tau\Phi_\tau)^{-1}\|}_{\leq O(1)} + \underbrace{\|P_\tau\|}_{\leq\bar P}\underbrace{\|(I+\tau\Phi_\tau)^{-1}\|}_{\leq O(1)}\underbrace{\|A\|}_{\leq O(1)}\big)\Big) 
    \]
    where we employed Item~\ref{lem:item:3} of Lemma~\ref{lem:neumannstuff}. This gives $\|K_\tau-\Krom\|\leq O(\|P_\tau -\Prom\|) + O(\tau)$, which is the desired bound.
\end{proof}
We will now only need to formulate a bound on $\|P_\tau - \Prom\|$ as a function of $\tau$. We achieve this by analyzing the solution to the DARE as~$\tau$ gets close to zero. To that end, define the scaled DARE residual operator as
\begin{equation}\label{eq:DAREres}
     F(\tau,X)\;:=\; \frac{1}{\tau}\left(A_\tau^{\top} X A_\tau + Q_\tau - A_\tau^{\top} X B_\tau
      \bigl(R_\tau + B_\tau^{\top} X B_\tau\bigr)^{-1} B_\tau^{\top} X A_\tau - X \right)
\end{equation}
for all $\tau> 0$. 
We note that $F(\tau,X)$ is jointly smooth in $(\tau,X)$, and satisfies $F(\tau,P_\tau)=0$ for all $\tau > 0$.  We further recall that by a standard limiting argument we have that for $P_0=\Prom$, the residual~\eqref{eq:DAREres} gives: $F(0,P_0):=\lim_{\tau\downarrow 0}F(\tau,P_\tau)=F(0,\Prom)=0$. This follows directly from noting that the zeroth order terms of $F(\tau,\Prom)$ are precisely the CARE for which~$\Prom$ is the unique stabilizing solution. The following result uses this observation to analyze the dependence of~$F$ on~$\tau$.

\begin{lemma}\label{lem:ftaup}
    Let $\Prom$ be the unique, positive definite solution to the CARE, corresponding to~\eqref{eq:modredproblem}. For all $0\leq\tau\leq\bar\tau$ and~$\|P_\tau\|\leq \bar P$, we have that
    \[
    F(\tau,\Prom) = \tau\Gamma_1(\Prom)+\tau^2\Gamma_2(\Prom),
    \]
    where $\|\Gamma_i(\Prom)\|\leq O(1)$, $i=1,2.$
\end{lemma}
\begin{proof}
Algebraic expansion of $F(\tau,\Prom)$ gives
\begin{multline*}
    F(\tau,\Prom ) = A^\top \Prom + \Prom A + Q + \tau A^\top \Prom  A  \\ - \Prom B(R+\tau B^\top \Prom  B)^{-1}B^\top \Prom  - \tau A^\top \Prom  B (R+\tau B^\top \Prom B)^{-1}B^\top \Prom  \\ - \tau \Prom  B (R+\tau B^\top \Prom  B)^{-1}B^\top \Prom A - \tau^2 A^\top \Prom B (R+\tau B^\top \Prom B)^{-1}B^\top \Prom A.
\end{multline*}
Using Item~\ref{lem:item:2} in Lemma~\ref{lem:neumannstuff}, we can write
\[
    F(\tau,\Prom ) = A^\top \Prom  + \Prom A - \Prom BR^{-1}B^\top \Prom  + Q + \tau\Gamma_1(\Prom ) + \tau^2\Gamma_2(P),
\]
where
\begin{align*}
    \Gamma_1(\Prom ) & = A^\top \Prom  A - A^\top \Prom B (R+\tau B^\top \Prom B)^{-1}B^\top \Prom - \Prom B (R+\tau B^\top \Prom B)^{-1}B^\top\Prom A \\ & \qquad - \Prom BR^{-1}B^\top \Prom B (R+\tau B^\top\Prom B)^{-1}B^\top \Prom  \\
    \Gamma_2(\Prom ) & = -A^\top\Prom B (R+\tau B^\top\Prom B)^{-1}B^\top\Prom A
\end{align*}
We recall that the $O(1)$ term vanishes as it is precisely the CARE that $\Prom$ satisfies. Finally, Lemma~\ref{lem:neumannstuff} and the results of Prop.~\ref{prop:k-bounds} allow us to conclude that for small enough $\bar\tau$, $(R+\tau B^\top\Prom B)^{-1}\leq O(1)$ for all $0\leq\tau\leq\bar\tau$. Hence, we conclude~$\|\Gamma_i(\Prom)\|\leq O(1)$, $i=1,2$.
\end{proof}
The results of this lemma can now be used to formulate a bound on the directional derivative of~\eqref{eq:DAREres}. This bound, which we derive in the lemma below, allows us to finally formulate the bound on~$\|P_\tau-\Prom\|$.

Let $A_\textsc{cl}:=A_\textsc{r} + B_\textsc{r}\Krom$ and recall that $\alpha_i:=-\max_j\Re\{\lambda_j(A_\textsc{cl})\}$ for the reduced order model~$M_i$.
Then, a direct computation shows that the partial Fr{\'e}chet derivative of $F(\tau,X)$ with respect to $X$ in direction $H$ is given by
 \begin{multline}\label{eq:frechet}
 \DXF{\tau}{X}{H}
=\;
\frac{1}{\tau}\left(A_\tau^{\top} H A_\tau - H - \bigl(A_\tau^{\top} H B_\tau\bigr)
     \Sig^{-1}
     \bigl(B_\tau^{\top} X A_\tau \bigr) - \bigl(A_\tau^{\top} X B_\tau \bigr)
     \Sig^{-1}
     \bigl(B_\tau^{\top} H A_\tau\bigr)\right.\\ 
     \left. + \bigl(A_\tau^{\top} X B_\tau \bigr)
     \Sig^{-1}
     \bigl(B_\tau^{\top} H B_\tau\bigr)
     \Sig^{-1}
     \bigl(B_\tau^{\top} X A_\tau \bigr)\right),  
 \end{multline}
 where $\Sig := R_\tau + B_\tau^{\top} X B_\tau $.
\begin{lemma}\label{lem:lyap-inv}
 $\DXF{0}{\Prom}{H}$ is the Lyapunov operator induced by the closed-loop LQR system $A_{\textsc{cl}}$, i.e.,
 \begin{equation}\label{eq:lyap-op}
 \DXF{0}{\Prom}{H} = A_{\textsc{cl}}^\top H + HA_{\textsc{cl}}=:\Lcal(H).
 \end{equation}
 Furthermore, $\Lcal$ is invertible and satisfies $\|\Lcal^{-1}\|\leq O(\frac{1}{\alpha_i}).$%
\end{lemma}
\begin{proof}
Equation~\eqref{eq:lyap-op} follows from evaluating~\eqref{eq:frechet} as $\tau\searrow 0$ and recalling that $\Krom = -R^{-1}B^\top \Prom$.
As $A_{\textsc{cl}}$ is stable, $\Lcal$ is invertible, with 
\[
 \Lcal^{-1}(S)=\int_{0}^\infty e^{A_{\textsc{cl}}^\top s}Se^{A_{\textsc{cl}}s}\mathrm{d}s.
\]
   We recall that by Gelfand's formula, there exists some $C_{\textsc{cl}}>0$ such that $||e^{A_{\textsc{cl}}s}||\leq C_{\textsc{cl}}e^{-\alpha_i s}$.  Hence,
\[  \| \Lcal^{-1}(S)\|\leq \|S\|C_{\textsc{cl}}^2\int_0^\infty e^{-2\alpha_i s}\mathrm{d}s=\|S\|\frac{C_{\textsc{cl}}^2}{2\alpha_i},
\]
 and $\|\Lcal^{-1}\|\leq \frac{C_{\textsc{cl}}^2}{2\alpha_i}\leq O(\tfrac{1}{\alpha_i}).$
\end{proof}
These results allow us to guarantee the following bound on $\|P_\tau - \Prom\|$.

\begin{proposition}\label{prop:p-bounds}
    There exists $\tau'>0$ such that for all $0\leq \tau\leq \tau'$, $\|P_\tau - \Prom\|\leq O(\tau).$
\end{proposition}
\begin{proof}
    Recall that $F(\tau,P_\tau) = 0$. Using the definition of $F$ and Lemma~\ref{lem:lyap-inv}, we can now invoke the implicit function theorem to define a unique function $\varphi$ with $\varphi(\tau)=P_\tau$ and $F(\tau,\varphi(\tau)) = 0$ for all $0\leq\tau\leq\tau'$. Taking the Taylor approximation of $\varphi$ allows us to conclude that there exists a $\tau'>0$ such that $P_\tau = \Prom + [D_X F(0,\Prom)]^{-1}\partial_\tau F(0,\Prom)\tau + O(\tau^2)=\Prom + \Lcal^{-1}(\Gamma_1(\Prom))\tau+O(\tau^2)$ for all $0\leq \tau \leq \tau',$ i.e., such that $\|P_\tau-\Prom\|\leq \|\Lcal^{-1}\|\|\Gamma_1(\Prom)\|\tau + O(\tau^2)\leq O(\tau)$.
\end{proof}
We can now finally combine Prop.~\ref{prop:k-bounds} and Prop.~\ref{prop:p-bounds} to obtain the bound on $\|K_\tau - \Krom\|$ as a function of $\tau$.

\begin{proposition}\label{prop:final-kbounds}
    There exists a $\tau_{\max}>0$ such that for all $0\leq\tau\leq \tau_{\max}$ we have that
    $$
    \|K_\tau-\Krom\|\leq O(\tau).
    $$
\end{proposition}
\begin{proof}
    By Prop.~\ref{prop:p-bounds}, there exists some $\tau'>0$ such that $\|P_\tau-\Prom\|\leq O(\tau)$ for all $0<\tau\leq\tau'$.  It follows that for sufficiently small $\tau_{\max}:=\min\{\tau', \bar{\tau}\}$, we can guarantee that $\|P_\tau\|=\|\Prom + (P_\tau-\Prom)\|\leq \bar P$ for all $\tau\in(0,\tau_{\max}]$ and hence the assumptions of Prop.~\ref{prop:k-bounds} hold, from which the result follows immediately.
\end{proof}

\subsubsection{Bounding {\normalfont$\|\Krom\|\|x_t-\xdt\|$ and $\|\xdt\|$}}\label{sss:boundingZOH}
Our goal is to bound $\|\hat\pi_\textsc{zoh}(\tau,x_t)-\hat{\pi}_{M_i}(x_t)\|\leq\|\Krom\|\|x_t-\xdt\|$ and $\|\xdt\|$ as a function of~$\tau$ and~$\|x_t\|$. 
Note that $\|\Krom\|\leq O(1)$. %
Hence, we will only need to bound $\|x_t-\xdt\|$ and $\|\xdt\|$. 

As discussed in \S\ref{ss:lqrdesignapprox}, we will need to analyze these quantities with respect to the states of the optimal and perturbed system. In the following we show that for both cases, $\|x_t-\xdt\|\leq O(\tau)$ and $\|\xdt\|\leq O(1)$. %

\paragraph{Bound under the states of the optimal system}
The state evolution of the optimal system is governed by $u_t=K_\star x^\star_t$, i.e., $\dot{x}^\star_t = A_\textsc{lqr}x^\star_t$. Recall that $\alpha_0:=-\max_j\Re\{\lambda_j(A_\textsc{lqr})\}$ denotes the negative spectral abscissa of the optimal closed-loop. The following result shows that for this case, we have $\|x_t^\star-\xdt^\star\|\leq O(\tau)$ and $\|\xdt^\star\|\leq O(1)$: 
\begin{proposition}\label{prop:boundxstar}
    Let $\dot x^\star_t=A_{\textsc{lqr}}x^\star_t$. Then, for small enough $\tau \leq \tau_\textsc{zoh}$,
    \begin{equation}\label{eq:ZOH-xstar}
        \|x^\star_t-\xdt^\star\|\leq O(\tau) \|x_t^\star\|,
    \end{equation}
    and
    \begin{equation}\label{eq:zoh-xstar-1}
        \|\xdt^\star\|\leq O(1)\|x^\star_t\|.
    \end{equation}
\end{proposition}
\begin{proof}
    For $x^\star_t$ generated by $\dot{x}^\star = A_{\textsc{lqr}}x^\star$, we have that $x^\star_t=e^{A_{\textsc{lqr}}(t-\dt)}\xdt^\star$ (recall that $\dt:=\tau\floor{\frac{t}{\tau}}$) on a $\tau$-long sampling interval. Hence, $\xdt^\star=e^{-A_{\textsc{lqr}}(t-\dt)}x^\star_t$ and $x^\star_t-\xdt^\star=(I-e^{-A_{\textsc{lqr}}(t-\dt)})x^\star_t$. It therefore suffices to bound $\|I-e^{-A_{\textsc{lqr}}s}\|$ for all $0\leq s\leq\tau\leq\tau_\textsc{zoh}$. We write $I-e^{-A_{\textsc{lqr}}s}=e^{-A_{\textsc{lqr}}s}(e^{A_{\textsc{lqr}}s}-I)=(I+(e^{A_{\textsc{lqr}}s}-I))^{-1}(e^{A_{\textsc{lqr}}s}-I)$, and first bound $\|e^{A_{\textsc{lqr}}s}-I\|.$  To that end, note that
\[ e^{A_{\textsc{lqr}}s}-I=\int_{0}^s A_{\textsc{lqr}}e^{A_{\textsc{lqr}}\theta}\mathrm{d}\theta, \]
so that for all $0\leq s\leq\tau\leq\tau_\textsc{zoh}$, we have
\begin{multline*}
\|e^{A_{\textsc{lqr}}s}-I\|=\norm{\int_{0}^s A_{\textsc{lqr}}e^{A_{\textsc{lqr}}\theta}\mathrm{d}\theta}\leq\norm{A_{\textsc{lqr}}}C_{\textsc{lqr}}\int_{0}^se^{-\alpha_0\theta}\mathrm{d}\theta=\norm{A_{\textsc{lqr}}}C_{\textsc{lqr}}\frac{1-e^{-\alpha_0 s}}{\alpha_0}\\\leq \norm{A_{\textsc{lqr}}}C_{\textsc{lqr}}\frac{1-e^{-\alpha_0 \tau}}{\alpha_0} \leq (\norm{A_{\textsc{lqr}}}C_{\textsc{lqr}})\tau,
\end{multline*}
where we used Gelfand's formula.
It therefore follows immediately that for small enough $\tau_\textsc{zoh}\leq\tau_{\max}$, any $0\leq \tau\leq\tau_\textsc{zoh}$ will ensure that $(\norm{A_{\textsc{lqr}}}C_{\textsc{lqr}})\tau\leq\tfrac{1}{2}$, and thus $(I+(e^{A_{\textsc{lqr}}s}-I))^{-1}\leq O(1)$ and $\|e^{-A_{\textsc{lqr}}s}\|\leq 2$ for all $s\leq \tau$. This provides us with $\|\xdt^\star\|\leq O(1)\|x_t^\star\|$ and $\|x^\star_t-\xdt^\star\|\leq O(\tau) \|x_t^\star\|$. 
\end{proof}

\paragraph{Bound under the states of the perturbed system}
Considering the order of approximations in \S\ref{ss:lqrdesignapprox}, the state evolution of the perturbed system is governed by a state-feedback policy with a ZOH input, i.e., $u_t={K}\hat{x}_{\dt}$. Note that for this case, it does not matter what~${K}$ is. In the proposition that follows, we use the fact that under the policy~$u_t={K}\hat{x}_{\dt}$ the input is constant over the time-interval $\mc{I}:=[\dt, \dt+\tau]$. As highlighted in \S\ref{ss:baselinelqr}, we assume here that~$A$ is invertible. As will be evident from the proof, this assumption can easily be lifted using the Jordan form of~$A$. We are now able to state the bounds on $\|\hat x_t-\hat{x}_{\dt}\|$ and $\|\hat{x}_{\dt}\|$ under policy~$u_t={K}\hat{x}_{\dt}$.

\begin{proposition}\label{prop:boundKtau-xi}
    For all $\tau\leq \hat\tau$, we have
    \[
    \|\hat x_t-\hat{x}_{\dt}\|\leq O(\tau)\|\hat x_t\|
    \]
    and 
    \[
    \|\hat{x}_{\dt}\|\leq O(1)\|\hat x_t\|
    \]
    for $\hat x_t$ generated by $u_t=K \hat{x}_{\dt}$.
\end{proposition}
\begin{proof}
Note that during interval $\mc{I}$, we have that $\dot{\hat{x}}_t = A\hat x_t+B K \hat{x}_{\dt}$.  It follows immediately that for $t\in \mc{I}$, we have $\hat x_t=G(t-\dt)\hat{x}_{\dt}$ where
\[ 
G(s)=e^{As}+\int_0^se^{A(s-\theta)}\mathrm{d}\theta BK=e^{As} +(e^{As}-I)A^{-1}BK.
\]
Assuming that $G^{-1}(t-\dt)$ exists for $t\in \mc{I}$, we may then write $\hat{x}_{\dt} = G^{-1}(t-\dt)\hat x_t$ to conclude that
\[
\hat x_t-\hat{x}_{\dt} = \big(I-G^{-1}(t-\dt)\big)\hat x_t=\big(G(t-\dt)-I\big)G^{-1}(t-\dt)\hat x_t.
\]
Hence,
\[\|\hat x_t-\hat{x}_{\dt}\|\leq\|G(t-\dt)-I\|\|G^{-1}(t-\dt)\|\|\hat x_t\|,\] 
and it suffices to bound the two operator norms to reach our desired result. Let $C_A>0$, $\alpha_A\in \R$ be such that $\|e^{As}\|\leq C_A e^{\alpha_A s}$. To that end, first observe that 
\[
    G(s)-I=(e^{As}-I)(I + A^{-1}BK)=\int_0^s e^{A\theta}\mr{d}\theta A (I + A^{-1}BK)=\int_0^s e^{A\theta}\mr{d}\theta(A+B K),
\]
and we then have for all $s\in[0,\tau]$
\[
    \|G(s)-I\|\leq C_A\int_{0}^se^{\alpha_A \theta}\mr{d}\theta \|A+BK\|=C_A\|A+BK\|\frac{|e^{\alpha_A s}-1|}{|\alpha_A|}\leq C_A\|A+BK\|\frac{e^{|\alpha_A \tau|}-1}{|\alpha_A|}.
\]
Finally, for small enough $\hat\tau$, we have $|\alpha_A\tau|\leq \log 2$, so that by the mean value theorem $e^{|\alpha_A\tau|}-1\leq 2|\alpha_A|\tau$ and we conclude that
\begin{equation}\label{eq:gs-I-bound}
    \|G(s)-I\|\leq O(\tau)
\end{equation}
for all $s\in [0,\tau]$.
To compute $\|G^{-1}(s)\|$, we observe that $G^{-1}(s) = (I - (I-G(s))^{-1}$, and thus, so long as $\|G(s)-I\|<1$, which we can always guarantee for sufficiently small $\tau$, we have that
\[ \|G^{-1}(s)\|\leq O(1). \]
Thus, under the assumptions made above, %
we have that for all $\tau\leq\hat\tau$, $\|\hat x_t-\hat{x}_{\dt}\|\leq O(\tau) \|\hat x_t\|$ and $\|\hat{x}_{\dt}\|\leq O(1)\|\hat x_t\|$ for all $\hat x_t$ generated by $u_t = K \hat{x}_{\dt}$.
\end{proof}

\subsubsection{Bounding {\normalfont$\|\Krom\|\|x_t-\xdt\|+\|K_\tau-\Krom\|\|\xdt\|$}}
Let $\tau_\textsc{dt} := \min\{\tau_\textsc{zoh},\tau_{\max},\hat\tau\}$. Combining the results above yields the following bound:
\begin{equation*}
    \|(\hat\pi_\tau-\hat{\pi}_{M_i})\|= \|\Krom\|\|x_t-\xdt\|+\|K_\tau-\Krom\|\|\xdt\|\leq O(\tau)\|x_t\|.
\end{equation*}
for states~$x_t$ generated by either the continuous policy $u_t=K_\star x_t$ or the ZOH policy $u_t = K\xdt$ for all $0\leq\tau\leq\tau_\textsc{dt}$.

\subsection{Finite Horizon Approximation}\label{appendix:fhorizon}
For this approximation, we want to characterize the following design-dependent sector bound, which we can write as
\[ \|(\hat\pi_T-\hat\pi_\tau)(x_t)\| = \|(K_{T,\tau}-K_\tau)\xdt\| \leq \|K_{T,\tau}-K_\tau\|\|\xdt\|.\]
From Propositions~\ref{prop:boundxstar} and~\ref{prop:boundKtau-xi}, we have that $\|\xdt\|\leq O(1)\|x_t\|$ for small enough $\tau\leq\tau_\textsc{dt}$, for states generated by either the optimal or perturbed system. Thus, we only need to characterize $\|K_{T,\tau}-K_\tau\|$ in terms of $T,\tau$. Recall that we defined $N:=\floor{\tfrac{T}{\tau}}$.

For a fixed sampling time $\tau$, denote by $P^{(N)}_\tau$ the value computed by running the DARR for $N$ iterations, i.e., by computing $P^{(N)}_\tau:=P_{\tau,0}$, where  $P_{\tau,N}=Q_N$ is the terminal cost, and
$$
P_{\tau,k} = A_\tau^{\!\top} P_{\tau,k+1} A_\tau + Q_\tau
-
A_\tau^{\!\top} P_{\tau,k+1} B_\tau
      \bigl(R_\tau + B_\tau^{\!\top} P_{\tau,k+1} B_\tau\bigr)^{-1}
      B_\tau^{\!\top} P_{\tau,k+1} A_\tau 
$$
for $k=0,\dots,N-1.$  Further, define the map $K(X,\tau):=-(R_\tau+B_\tau^\top X B_\tau)^{-1}B_\tau^\top X A_\tau$.  Our goal in this section is to bound $\|K(P^{(N)}_\tau,\tau)-K(P_\tau,\tau)\|$ as a function of $N$ for all $\tau\leq\tau^\star$. We first bound this quantity as a function of $\|P^{(N)}_\tau-P_\tau\|$:
\begin{lemma}\label{lem:boundKTtau}
    Let $U$ be a compact neighborhood of $P_\tau$ contained within the positive semi-definite cone.  Then, for all $X\in U$ and all $\tau\in(0,\breve\tau]$, we have that
    \[
    \|K(X,\tau)-K(P_\tau,\tau)\|\leq O(\|X-P_\tau\|).
    \]
\end{lemma}
\begin{proof}
    We first observe that for a fixed $\tau$, the map $K(X,\tau)$ is locally Lipschitz in $X$ in any compact neighborhood $U$ around $P_\tau$ contained within the positive semi-definite cone.  This follows from
\[ \|K(X,\tau)-K(P_\tau,\tau)\|\leq \sup_{\gamma\in[0,1]} \|D_X K(X + \gamma P_\tau,\tau)\|\|X-P_\tau\|, \]
and $D_X K(X,\tau)$ being continuous in $X\succeq 0$, and hence the supreme is finite for all $X\in U$.  We denote by $L(\tau):= \sup_{\gamma\in[0,1]} \|D_X K(X + \gamma P_\tau,\tau)\|$ and note that $L_P:=\sup_{\tau\in(0,\breve\tau]} L(\tau)$ can be verified to be finite through routine algebra. Hence, $ \|K(X,\tau)-K(P_\tau,\tau)\|\leq L_P \|X-P_\tau\| \equiv O(\|X-P_\tau\|).$
\end{proof}
Next, we turn to quantifying the rate at which the sequence $\{P_\tau^{N}\}$ converges towards $P_\tau$ as a function of the number of iterations $N$.  Standard convergence rate analysis of the Riccati difference equation~\cite{bitmead1991riccati,hager1976convergence} shows that under stabilizability and detectability of $(A_\tau,B_\tau,Q_\tau^{1/2})$, the sequence $\{P_\tau^{N}\}$ converges geometrically in the operator norm  to $P_\tau$, which is summarized in the following proposition:
\begin{proposition}
	Let $N:=\floor{T/\tau}$. Then, for all $\tau\leq \tau_\textsc{fh}$,  sufficiently large $N_0$, and sufficiently small~$\epsilon$, such that $N\geq N_0$ and $\|Q_N-P_\tau\|\leq\epsilon$, there exists a polynomial $C_T(\tau)$ such that
	\begin{equation}\label{eq:boundexpconverg}
		\|K(P_\tau^{(N)})-K(P_\tau)\|\leq C_T(\tau) e^{-\alpha_i T}.
	\end{equation}
\end{proposition}
\begin{proof}
Let $E_k:= P_{\tau,k}-P_\tau$ and $A_\textsc{dlqr}:= A_\tau + B_\tau K_\tau$. As $A_\textsc{dlqr}$ is stable, its spectral radius satisfies $\rho(A_\textsc{dlqr})<1$, and thus for the linear (Lyapunov) operator $\Lcal_{A_\textsc{dlqr}}(E_{k+1}):=A_\textsc{dlqr}^\top E_{k+1}A_\textsc{dlqr}$ we have that: $\rho(\Lcal_{A_\textsc{dlqr}})=\rho(A_\textsc{dlqr})^2<1$. Writing out $E_k$ gives
\begin{equation}\label{eq:app:conv}
	E_k = A_\textsc{dlqr}^\top E_{k+1} A_\textsc{dlqr} + \mc{R}(E_{k+1}),
\end{equation}
where $\mc{R}(E_{k+1}):= -A_\textsc{dlqr}^\top E_{k+1} B_\tau(R_\tau+B^\top P_{\tau,k+1}B_\tau)^{-1} B_\tau^\top E_{k+1}A_\textsc{dlqr}$. Note that if $P_{\tau,k+1}\succ0$, $\mc{R}(E_{k+1})\preceq0$ for all $E_{k+1}$. Moreover, as $\mc{R}(E_{k+1})$ is quadratic in $E_{k+1}$, $\|\mc{R}(E_{k+1})\|\leq O(\|E_{k+1}\|^2)$, i.e., for small enough $\|E_{k+1}\|$, $\|\mc{R}(E_{k+1})\|\leq c\|E_{k+1}\|^2$. Let $\rho(A_\textsc{dlqr})^2\leq\gamma<1$ such that $\|\Lcal_{A_\textsc{dlqr}}(E_{k+1})\|\leq\gamma\|E_{k+1}\|$. Then, for $\|E_{k+1}\|\leq \epsilon$, where $c\epsilon<1-\gamma$, we have 
\[ \|E_k\|\leq \|\Lcal_{A_\textsc{dlqr}}(E_{k+1})\|+c\|E_{k+1}\|^2\leq (\gamma+c\epsilon)\|E_{k+1}\|.\]
Note that $\gamma$ can be chosen arbitrarily close to $\rho(A_\textsc{dlqr})^2$ for a given norm, and that $(\gamma+c\epsilon)<1$ by definition. Hence, due to backwards recursion, $\|P_\tau^{(N)}-P_\tau\|\leq (\gamma+c\epsilon)^N\|Q_N-P_\tau\|$, for the ball around $P_\tau$ defined by $\{X:\|X-P_\tau\|\leq\epsilon\}$. Thus, for large enough $N_0$ and small enough $\epsilon$, there exists a $C_0>0$ such that
\[\|P_\tau^{(N)}-P_\tau\|\leq C_0\rho(A_\textsc{dlqr})^{2N}\|Q_N-P_\tau\|.\] 
Next, for the reduced order model $M_i$, let $A_{\textsc{lqr}_i}:=A_{\textsc{r}}+B_{\textsc{r}}\Krom$, and recall that we defined $\alpha_i=-\max_j\Re\{\lambda_j(A_{\textsc{lqr}_i})\}$. Then, as $A_\textsc{dlqr}=I + \tau A_{\textsc{lqr}_i} + \tau B(K_\tau-\Krom) = I + \tau A_{\textsc{lqr}_i} +  O(\tau^2)$, cf. Prop.~\ref{prop:final-kbounds}. Using standard perturbation techniques, we conclude that the spectral radius of $A_\tau+B_\tau K_\tau$ satisfies $\rho(A_\tau + B_\tau K_\tau)=1-\alpha_i\tau+O(\tau^2)\leq 1-\frac{\alpha_i\tau}{2}$ for all $0\leq\tau\leq\tau_\textsc{fh}$ for small enough~$\tau_\textsc{fh}$. Therefore, $\rho(A_\textsc{dlqr})^{2N}\leq(1-\frac{\alpha_i\tau}{2})^{2N}\leq e^{-\alpha_i\tau N}\leq e^{\alpha_i\tau}e^{-\alpha_i T}$, where we use that $\tau N\leq T \leq\tau (N+1)$ and $e^{-\alpha_i\tau (N+1)}\leq e^{-\alpha_iT}$. Taking the Taylor series of~$e^{\alpha_i\tau}$ with respect to~$\tau$ gives a polynomial in~$\tau$, which, combined with $\epsilon C_0$, allows us to define the polynomial $C_T(\tau)$. Together with the result in Lemma~\ref{lem:boundKTtau}, we therefore conclude that~\eqref{eq:boundexpconverg} holds for all $\tau\leq \tau_\textsc{fh}$ with $\floor{T/\tau}$ sufficiently large and $\|Q_N-P_\tau\|$ sufficiently small.
\end{proof}

Combining all the above results, and defining $\tau_\star:=\min\{\tau_\textsc{dt}, \breve\tau, \tau_\textsc{fh}\}$ yields the design-dependent sector bound~\eqref{eq:pidecomposition_lqr}.

\subsection{Lipschitz Bounds on the Cost Function}\label{ass:costcond}

We now show that the conditions~\eqref{eq:cost_x} and~\eqref{eq:cost_pi} hold for the quadratic cost in~\eqref{eq:lqr-base}, and hence that we can construct the compute constraint receding horizon LQR DMP presented in Theorem~\ref{thm:dmp-lqr}. From the above derivations, we observe that the deployed policy~$\hat{\pi}(x)=K_{\tau,T_\star(\tau,M_i)}\xdt$ is \emph{linear} in~$x_t$: $\hat{\pi}(x_t)=\hat{K}(t-\dt,\tau,M_i)x_t$, where we lumped the effects of reduced-order modeling, ZOH discretization, and finite horizon approximation in the gain matrix function~$\hat{K}$. Moreover, under the assumption that~$\|(\hat\pi-\pi^\star)(x^\star)\|\leq L^\star(\tau,M_i)\|x^\star\|$, we have that
\[ \|\hat{K}(t-\dt,\tau,M_i)\|\leq \underbrace{\|\hat{K}(t-\dt,\tau,M_i)-K_\star\|}_{L^\star(\tau,M_i)}+\|K_\star\|\leq O(\|K_\star\|).\]
Hence, there exists a $B_K>0$ such that~$\|\hat{K}(t-\dt,\tau,M_i)\|\leq B_K\|K_\star\|$, which we will use in the following result.
\begin{proposition}\label{prop:boundingcost}
	Consider the quadratic cost $c(x,u) = x^\top Q x + u^\top R u$, and the stabilizing policies~$\hat\pi(x) = \hat K x$ and $\pi^\star(x) = K_\star x$. Define 
	\( \mc{X}:= \{x\in\R^\dx : \|x\|\leq B_X \}, \)
    and assume that $x_0\in\mc{X}$. Furthermore, assume that for $\|(\hat\pi-\pi^\star)(x^\star)\|\leq \bar{L}\|x^\star\|$, with $\bar L = \sup_{0<\tau\leq \tau_\star, i=0,\dots,3} L^\star(\tau,M_i)$, the small-gain condition~\eqref{eq:smallgaincondition} holds. Then,~\eqref{eq:cost_x} and~\eqref{eq:cost_pi} hold with
    \begin{equation}\label{eq:LxLpi}
        L_\mathrm{x}:= 3\sqrt{\cn(P_\star)}B_X\big(\|Q\| + 2B_K\|R\|\big), \quad L_\pi:= \sqrt{\cn(P_\star)}B_X\|R\|\big(2\|K_\star\|+\bar{L}\big).
    \end{equation}
\end{proposition}
\begin{proof}
    We first show that~\eqref{eq:cost_x} holds. By Cauchy-Schwarz, we have
\begin{multline*}
	|c(\hat x, \hat\pi(\hat x))-c(x^\star, \hat\pi(x^\star))| =   \big|\hat x^\top (Q + \hat K^\top R \hat K) \hat x - (x^\star)^\top (Q + \hat K^\top R \hat K) x^\star \big| \\ = \big|(\hat x-x^\star)^\top (Q + \hat K^\top R \hat K) \hat x - (x^\star)^\top (Q + \hat K^\top R \hat K) (\hat x-x^\star) \big|\\ \leq (\|Q\| + 2\|\hat K\|\|R\|)(\|\hat x\| + \|x^\star\|)\|\hat x-x^\star\|,
\end{multline*}
where we suppressed the $t$ subscript. 
As $x_0\in\mc{X}$, and the baseline closed-loop is stable with Lyapunov function $(x_t^\star)^\top P_\star x_t^\star$, we can use a simple Lyapunov argument to conclude that
\[ \|x_t^\star\|^2\leq\cn(P_\star)\|x_0^\star\|^2 = \cn(P_\star)\|x_0\|^2\leq\cn(P_\star)B_X^2, \] 
i.e., $\|x_t^\star\|\leq\sqrt{\cn(P_\star)}B_X$. Next, we note that $\|\hat x_t\|\leq \|\hat x_t-x_t^\star\| + \|x_t^\star\|$. Hence, under the small-gain and $\delta$ISS assumption, Props.~\ref{prop: small-gain},~\ref{prop:lqr-ediss}  give
\[ \|\hat x_t-x_t^\star\| \leq \sqrt{\cn(P_\star)}\|x_0^\star\|\leq \sqrt{\cn(P_\star)}B_X \implies  \|\hat x_t\|\leq 2\sqrt{\cn(P_\star)}B_X. \]
Defining~$L_\mathrm{x}$ as in \eqref{eq:LxLpi} gives the desired bound. Showing that~\eqref{eq:cost_pi} holds with $L_\pi$ as defined in~\eqref{eq:LxLpi} follows similar arguments. We have by Cauchy-Schwarz and the above reasoning that
\begin{multline*}
	|c(x^\star, \hat\pi(x^\star))-c(x^\star, \pi^\star(x^\star))| = \big|(x^\star)^\top (Q - Q + \hat K^\top R \hat K - K_\star^\top R  K_\star) x^\star\big|\\ = \big|((\hat K-K_\star)x^\star)^\top R (\hat K x^\star) + (K_\star x^\star)^\top R ((\hat K-K_\star)x^\star)\big|\leq \|x^\star\|\|R\|(\|K_\star\|+\|\hat K\|)\|(\hat\pi-\pi^\star)(x)\|,
\end{multline*}
where we again suppressed the $t$ subscript. 
Using the same arguments, we define~$L_\pi$ as in~\eqref{eq:LxLpi} to obtain the desired bound and complete the proof. %
\end{proof}

\subsection{Extended discussion on the numerical example}\label{appendix:example}
To improve reproducibility of our results, we provide some additional details on the implementation of the numerical example in~\S\ref{ss:numerical}. We will also provide extra plots, such as time-series simulations of the different deployed policies, and additional insights when hyperparameters in the DMP change.
\subsubsection{Implementation aspects}\label{asss:implement}
As mentioned in~\S\ref{ss:numerical}, we estimate the values for $C_i^\star,\hat{C}_i$ by fitting the functional form~\eqref{eq:pidecomposition_lqr} with 10 randomly selected points of~$\|(\hat\pi_T-\pi^\star)(x)\|$ for different values of~$\tau$. Using the above derivations, we can elaborate on how we exactly computed these values for a given~$\tau$. 

The optimal policy is~$\pi^\star(x_t) = K_\star x_t$, while the deployed policy is $\hat\pi_T(x_t) = K_{\tau,T_\star}\xdt$. We must evaluate the policies on trajectories of the baseline system, $x^\star_t$, and the deployed system, $\hat{x}_t$, where the difficulty is in evaluation of $\xdt$. Using Prop.~\ref{prop:boundKtau-xi}, we can write evaluation of $\xdt$ on the baseline or deployed system as:
\begin{align*}
	x^\star_{\dt}&=\left(e^{A(t-\dt)} +(e^{A(t-\dt)}-I)A^{-1}BK_\star\right)^{-1}{x}^\star_t=:G^{-1}_\star(t-\dt)x^\star_t \\ 
	\hat{x}_{\dt}&=\left(e^{A(t-\dt)} +(e^{A(t-\dt)}-I)A^{-1}BK_{\tau, T_\star}\right)^{-1}\hat{x}_t=:\hat G^{-1}(t-\dt)\hat x_t,
\end{align*}
under the assumption that $\hat G,G_\star$ are invertible, which holds for small enough~$\tau$. Hence, we can write for a given sampling time~$\tau$
\begin{align*}
	\|(\hat\pi_T-\pi^\star)(x^\star_t)\|=\|(K_{\tau, T_\star}G^{-1}_\star(t-\dt)-K_\star)(x^\star_t)\|&\leq \|K_{\tau, T_\star}G^{-1}_\star(\tau)-K_\star\|\|x_t^\star\|, \\
	\|(\hat\pi_T-\pi^\star)(\hat x_t)\|=\|(K_{\tau, T_\star}\hat G^{-1}(t-\dt)-K_\star)(\hat x_t)\|&\leq \|K_{\tau, T_\star}\hat G^{-1}(\tau)-K_\star\|\|\hat x_t\|,
\end{align*}
where we used the fact that $\max_{0\leq s\leq \tau}G^{-1}_\star=G_\star^{-1}(\tau)$ and $\max_{0\leq s\leq \tau}\hat G^{-1}=\hat G^{-1}(\tau)$. Using these explicit forms of the sector bounds, we can easily compute the norms for a given~$\tau$. In our implementation, we used the standard spectral matrix norm.

We used a similar approach to compute the (locally) optimal sampling time that minimizes $\|A_\textsc{lqr} x_t-(Ax_t + BK_{T_\star,\tau}\xdt)\|$ for all~$x_t$ for a given $M_i$. Minimizing this quantity is equivalent to minimizing the difference in response on the interval $[t,t+\tau]$. The closed-loop response of the baseline system on this interval is $x_{t+\tau}=e^{A_\textsc{lqr}\tau}x_t$, while the closed-loop response of the deployed system on this interval is $x_{t+\tau}=\hat{G}(\tau)x_t$. Hence, to compute the (locally) optimal sampling time we can simply compute $\min_{0\leq\tau}\|e^{A_\textsc{lqr}\tau}-\hat{G}(\tau)\|$, which can be obtained using nonlinear optimization.

\subsubsection{Time-series simulation}
\begin{figure}
\centering
\subfloat[$M_0$]{\includegraphics[width=0.49\linewidth]{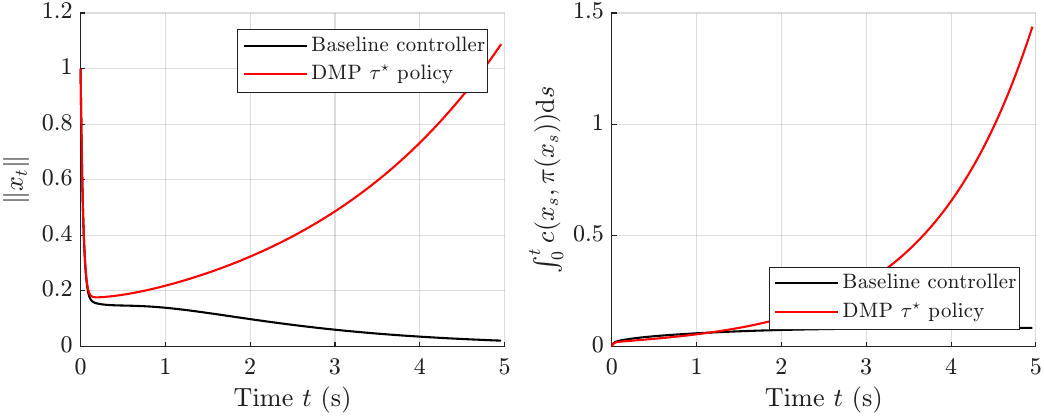}
	\label{fig:simulation_M0}}~\subfloat[$M_1$]{\includegraphics[width=0.49\linewidth]{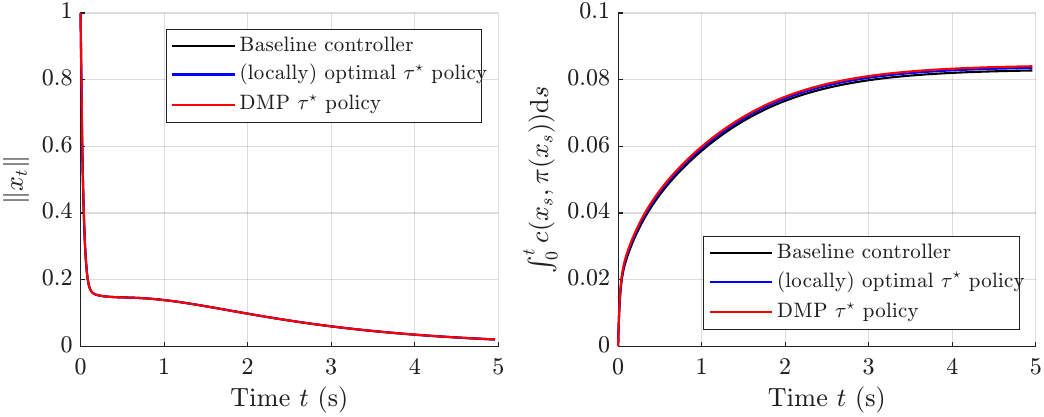}
	\label{fig:simulation_M1}}\\
\subfloat[$M_2$]{\includegraphics[width=0.49\linewidth]{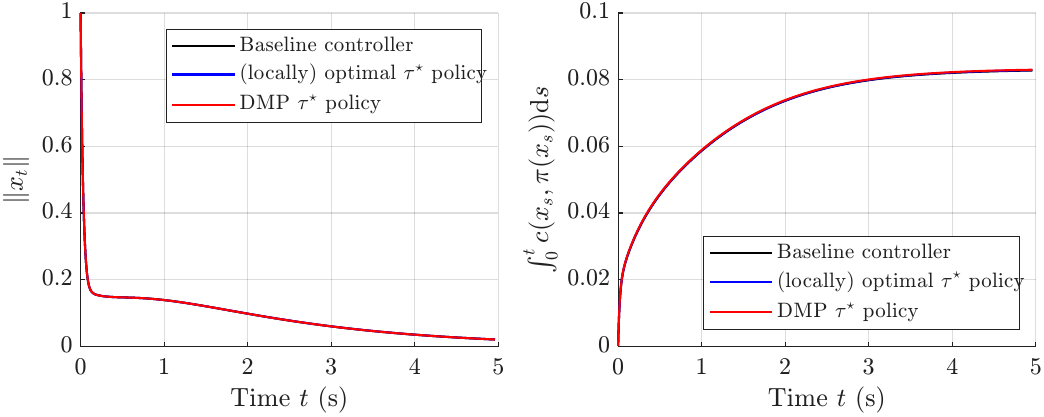}
	\label{fig:simulation_M2}}~\subfloat[$M_3$]{\includegraphics[width=0.49\linewidth]{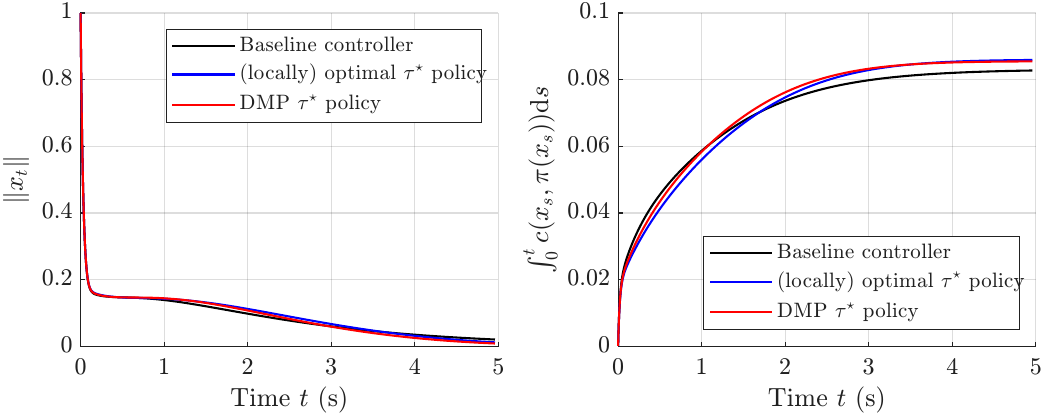}
	\label{fig:simulation_M3}}
\caption{Simulation of the baseline (\legendline{black}) and deployed closed-loop systems, where the selected sampling time~$\tau$ corresponds to either the locally optimal~$\tau$ (\legendline{blue}), obtained using nonlinear optimization, or the~$\tau$ that minimizes~\eqref{eq:lqr-dmp-tau-cost} (\legendline{red}). The left-hand plots show~$\|x_t\|_2$, with $x_t\in\R^{97}$. The right-hand plots shows the approximation of the closed-loop cumulative cost, i.e., $\int_0^t c(x_s,u_s)\mr{d}s$.}\label{fig:simulations}
\end{figure}
We simulate the baseline system together with the deployed system for the sampling times~$\tau$ that correspond to the minima of the fits in the left-hand plots in Fig.~\ref{fig:subopt}, as well as the (locally) optimal sampling time. The obtained values are listed in Table~\ref{tab:taus}.%
\begin{table}[t]
    \centering
    \caption{Values for the optimal sampling time obtained from the numerical example.}
    \label{tab:taus}
    \begin{tabular}{cccc}
    \hline Model & (locally) optimal $\tau^\star$ & $\tau^\star$ from DMP & Note\\\hline
         $M_0$ & No solution found & $1.0\cdot10^{-4}$ &  \eqref{eq:lqr-dmp-tau-sgc} not satisfied \\
$M_1$ & $2.0\cdot10^{-2}$ & $2.6\cdot10^{-2}$ &  \eqref{eq:lqr-dmp-tau-sgc} not satisfied \\
$M_2$ & $7.6\cdot10^{-3}$ & $1.1\cdot10^{-2}$ &   \\
$M_3$ & $9.3\cdot10^{-4}$ & $4.2\cdot10^{-3}$ &  \eqref{eq:lqr-dmp-tau-sgc} not satisfied \\\hline
    \end{tabular}
\end{table}
\begin{figure}[!hb]
	\centering
	\includegraphics[width=0.5\linewidth]{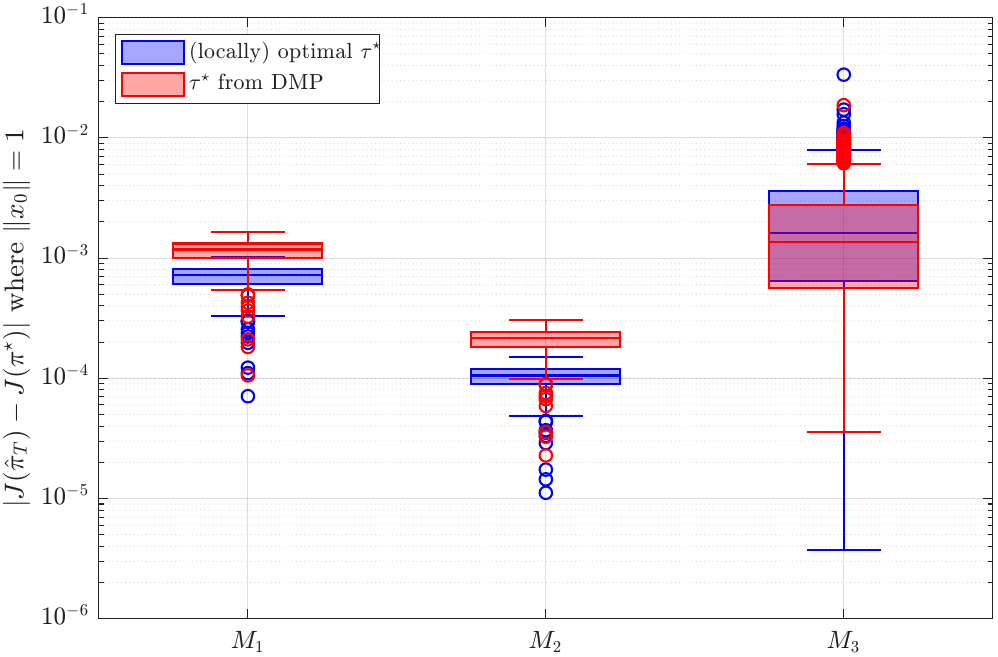}
	\caption{Box plot showing the distribution of the suboptimality gap of the deployed policies corresponding to $M_1,M_2,M_3$, based on a thousand 10-second simulations with initial conditions on the unit ball, i.e.,~$\|x_0\|=1$. The deployed policies are designed with a sampling time $\tau^\star$, obtained using nonlinear optimization or the DMP~\eqref{eq:lqr-dmp-tau}. Outliers in the distribution are indicated by $\circ$.}
	\label{fig:boxplot}
\end{figure}
For these sampling times, we have simulated the deployed closed-loop for $10$~seconds with a randomly generated $x_0$, with~$\|x_0\|_2=1$. Additionally, we simulate the baseline closed-loop system with the same initial condition for comparison. Note that the deployed closed-loop is the true system~\eqref{eq:lqr-base:model} controlled by the deployed policy~$\hat\pi_T(x_t)$, while the baseline closed-loop is the true system controlled with the baseline policy~$\pi^\star(x_t)$. Due to the large state dimension of the true system, we only plot~$\|x_t\|_2$. 
For the obtained response, we have also computed an approximation of the closed-loop cumulative cost, i.e., $\int_0^t c(x_s,\pi(x_s))\mr{d}s$, using a standard Riemann sum. The simulation results for the reduced-order models $M_0,\dots,M_3$ are shown in Fig.~\ref{fig:simulations}. As already highlighted in~\S\ref{ss:numerical}, we could not find a deployable policy for~$M_0$, which is confirmed by the simulation of the resulting policy obtained after making the discussed approximations: unstable closed-loop behavior. Carefully inspecting the difference between Fig.~\ref{fig:simulation_M1} and Fig.~\ref{fig:simulation_M2} allows us to conclude that the deployed policy obtained with~$M_2$ shows the closest resemblance to the baseline simulation. The simulation in Fig.~\ref{fig:simulation_M3} shows a significant difference in behavior compared to the baseline closed-loop system, i.e., the reduced-order model~$M_3$ is a too crude approximation of the true system. Finally, we also want to remark that these simulations demonstrate the conservativeness of~\eqref{eq:lqr-dmp-tau-sgc}: the sampling times that minimized~\eqref{eq:lqr-dmp-tau-cost} for $M_1$ and $M_3$ did not satisfy~\eqref{eq:lqr-dmp-tau-sgc}, while the resulting policies \emph{did} result in stable behavior.

To better show and compare the resulting suboptimality of the deployed controllers, we have simulated the aforementioned closed-loop systems for 1000 random initial conditions on the unit ball, i.e., $\|x_0\|_2=1$,  and then computed 
\[ \left|\int_0^{10} c(\hat{x}_t,\hat{\pi}_T(\hat{x}_t))\mr{d}t-\int_0^{10} c(x^\star_t,{\pi}^\star(x^\star_t))\mr{d}t \right|. \]
These simulations lead to a distribution of values of the suboptimality of the deployed policy. We have depicted this distribution for the different reduced order models in Fig.~\ref{fig:boxplot}.
This box plot clearly underlines the conclusion drawn in~\S\ref{ss:numerical}: the sampling time that minimizes the DMP for~$M_2$ provides the optimal tradeoff between the design parameters subject to the compute constraints.

\subsubsection{Improving functional form}
\begin{figure}[!b]
\centering
\subfloat[$M_1$]{\includegraphics[width=0.5\linewidth]{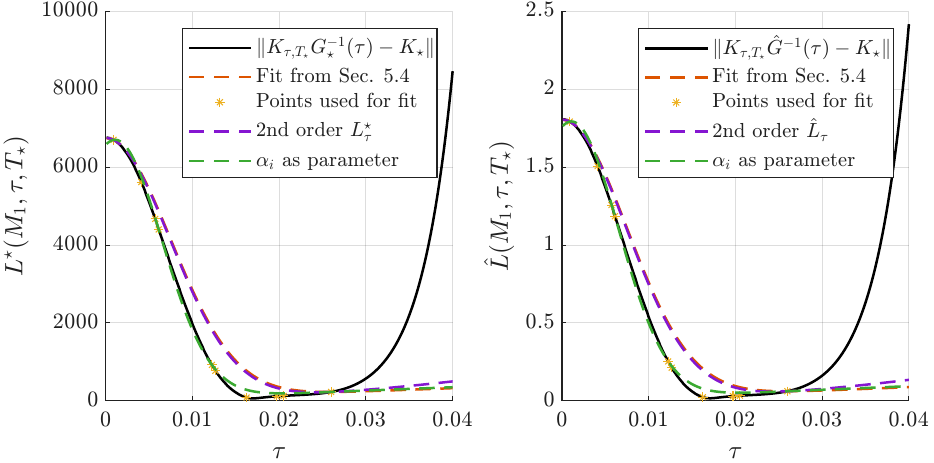}\label{fig:fitres_M1}}~
\subfloat[$M_2$]{\includegraphics[width=0.5\linewidth]{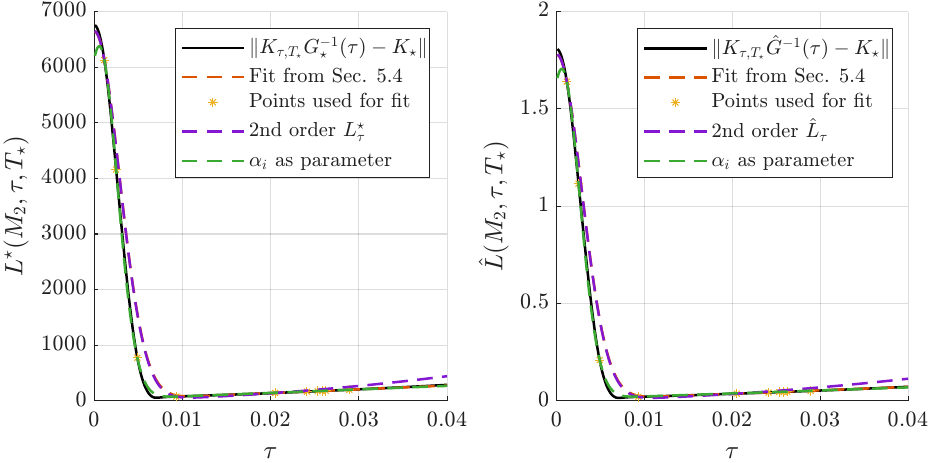}\label{fig:fitres_M2}}\\
\subfloat[$M_3$]{\includegraphics[width=0.5\linewidth]{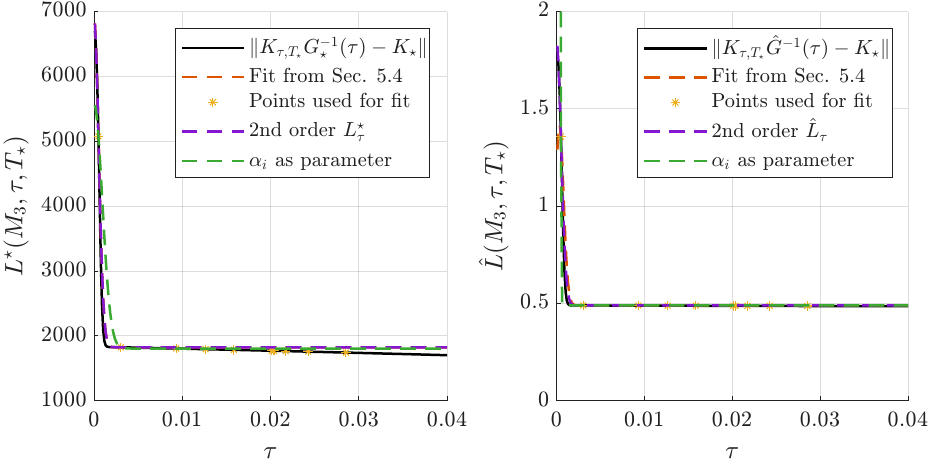}\label{fig:fitres_M3}}
\caption{Plots demonstrating the effect of choosing linear dependence for~$L_\tau$ on~$\tau$ and restricting the exponential rate to~$\alpha_i$ on the functional form~\eqref{eq:pidecomposition_lqr} (\legendline{morange}~\legendline{morange}). For the reduced-order models $M_1,M_2,M_3$, we fitted here the obtained functional form again, now with quadratic dependence for~$L_\tau$ on~$\tau$ (\legendline{mpurple}~\legendline{mpurple}), and~$\alpha_i$ as a fittable parameter (\legendline{mgreen}~\legendline{mgreen}). We used the same 10 randomly selected points (\textcolor{myellow}{$\ast$}) from the computed ``true'' suboptimality gap of the deployed policy (\legendline{black}).}\label{fig:fitres}
\end{figure}
To generate the results in~\S\ref{ss:numerical}, we chose the form
\[ \|(\hat\pi_T-\pi^\star)(x_t)\|\leq  \left(C_1\tau + (C_2 + C_3\tau)e^{-\alpha_i T} + L_\Mcal(M_i)\right)\|x_t\|,  \]
i.e., we chose to restrict $L_\tau$ in~\eqref{eq:appdt:bound} to be linearly dependent on~$\tau$, and the exponential rate to be expressed in terms of~$\alpha_i$. This inevitably affects the quality of the fit of the above functional form, which is the upper bound on $\|(\hat\pi_T-\pi^\star)(x_t)\|$ for $0\leq\tau\leq\tau_\star$. In this section, we study the effect of these choices by considering instead 2\textsuperscript{nd} order dependence on~$\tau$ for~$L_\tau$ and considering~$\alpha_i$ as a fittable parameter in the functional form.

Fig.~\ref{fig:fitres} shows the resulting fits with these modifications in the DMPs for $M_1,M_2,M_3$. Specifically, for a given $M_i$, it shows the ``true'' suboptimality gap of the deployed policy in black, which we computed as discussed in~\S\ref{asss:implement}. The fits presented in Fig.~\ref{fig:subopt} are plotted in Fig.~\ref{fig:fitres} with same color for convenience. Using the same randomly selected points as in~\S\ref{ss:numerical} (indicated by yellow asterisks), we fitted the functional form
\( (C_{1,1}\tau + C_{1,2}\tau^2) + (C_{2} + C_{3}\tau)e^{-\alpha_i T_\star(\tau,M_i)} + L_\Mcal(M_i), \)
with respect to $\|x^\star_t\|$ and $\|\hat{x}_t\|$, indicated by the purple dashed lines. The green dashed lines indicate the fit of the functional form
\( C_{1}\tau + (C_{2} + C_{3}\tau)e^{-\alpha_{\textsc{fit},i} T_\star(\tau,M_i)} + L_\Mcal(M_i),\)
where $\alpha_{\textsc{fit},i}$ is a value to be estimated. The results in Fig.~\ref{fig:fitres} show that the quality of the fit is most affected by the decay rate~$\alpha_i$. We have also increased the order for~$C_T(\tau)$ (not shown), which did not noticeably improve the fit.

\subsubsection{Improved computational resources}
The results in~\S\ref{ss:numerical} consider the implementation of the deployed policy on a simple Arduino Leonardo. Upgrading the hardware of the system design results in better computational resources, i.e., a smaller~$\tau_g$. The impact of the larger computational budget is visualized in Fig.~\ref{fig:subopt_multhz} on Page~\pageref{fig:subopt_multhz}, where for every following subplot, we doubled the number of flops per second, i.e., halved~$\tau_g$.

We can draw an important conclusion from the consecutive subplots in Fig.~\ref{fig:subopt_multhz}. It is easy to see that more computational resources directly enlarge the feasibility region of the small gain condition~\eqref{eq:lqr-dmp-tau-sgc}. Moreover, performing the same analysis for $M_1$ (not shown) gives that, despite not feasible with an Arduino Leonardo, the DMP \emph{is} feasible when, e.g., $1/\tau_g=128$ MHz. Hence, increasing the computational budget will result in a larger range of feasible solutions across different design dimensions (e.g., state dimension of the reduced order model, sampling time, and related horizon length). 

\begin{figure}[t]
\centering
\subfloat[$1/\tau_g=16$ MHz (i.e., Fig.~\ref{fig:suboptrho})]{\includegraphics[width=0.49\linewidth]{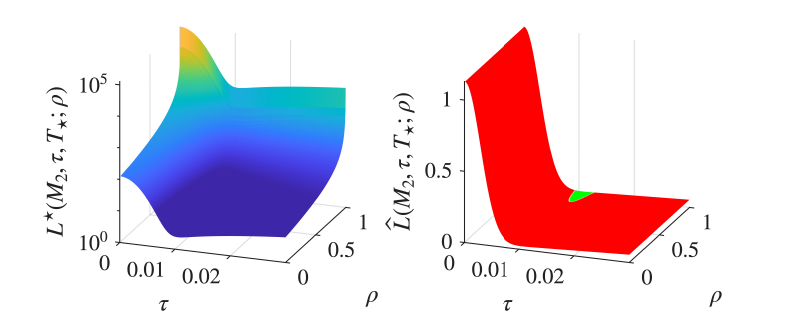}
	\label{fig:suboptrho_16mhz}}
\subfloat[$1/\tau_g=32$ MHz]{\includegraphics[width=0.49\linewidth]{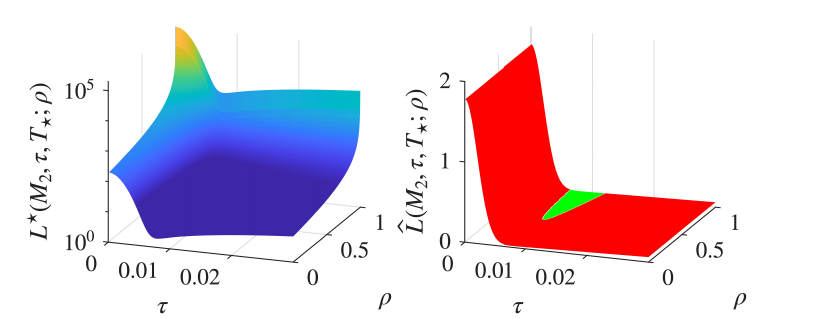}
	\label{fig:suboptrho_32mhz}}\\
\subfloat[$1/\tau_g=64$ MHz]{\includegraphics[width=0.49\linewidth]{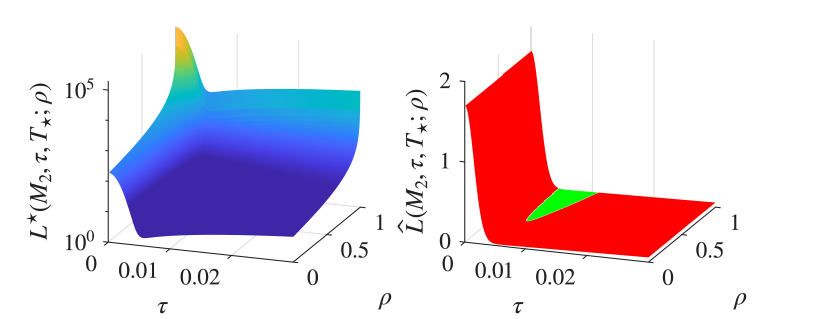}
	\label{fig:suboptrho_64mhz}}
\subfloat[$1/\tau_g=128$ MHz]{\includegraphics[width=0.49\linewidth]{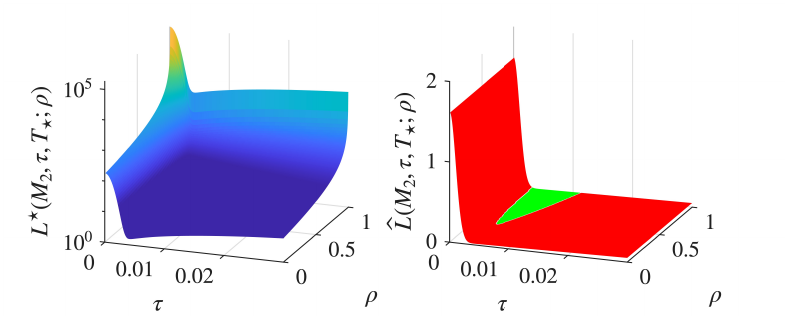}
	\label{fig:suboptrho_128mhz}}
\caption{Evolution of $L^\star$ and $\hat{L}$ when the computational resources grow. We showed here the computed values of the design-dependent sector bound over a grid of $\tau\times\rho$ for the optimal and deployed systems for system~$M_2$ and a given $\tau_g$. As in Fig.~\ref{fig:suboptrho}, the region for which the small-gain condition~\eqref{eq:lqr-dmp-tau-sgc} is satisfied is indicated with green.}\label{fig:subopt_multhz}
\end{figure}

}

\end{document}

\else
	\documentclass[11pt]{article}
\def\arxivversion{1}
\usepackage{textcomp}
\usepackage{amsmath}
\usepackage{amssymb}
\usepackage{amsthm}
\usepackage{mathtools}
\usepackage{tikz}
\usepackage{pgfplots}
\usepackage{graphicx}
\usepackage{epstopdf}
\usepackage{wrapfig}
\usepackage{xcolor}
\usepackage{bbm}
\usepackage{makecell}
\usepackage{authblk}
\usepackage{fullpage,etoolbox}
\usepackage[shortlabels]{enumitem}
\usepackage[round,sort,compress]{natbib}
\usepackage{listings}
\usepackage[ruled,vlined,linesnumbered]{algorithm2e}
\usepackage{float}
\usepackage{cancel}
\usepackage{makecell}
\usepackage{balance}  %
\usepackage[colorlinks=true,citecolor=blue,linkcolor=blue,urlcolor=blue]{hyperref}

\usepackage{auxFiles/symbolDef}
\input{auxFiles/auxCommands}
\renewcommand{\qedsymbol}{$\blacksquare$}
\usepackage{subfig}
\usepackage[labelfont=bf]{caption}
\captionsetup[subfloat]{labelfont=bf}

\renewcommand{\cite}[2][]{\citep[#1]{#2}}

\newcommand{\arxivver}[2]{%
  \ifx\arxivversion\undefined%
    #1%
  \else%
    #2%
  \fi%
}

\newcommand{\comment}[1]{}
\newcommand{\todo}[1]{{\color{red}#1}}
\newcommand{\remove}[1]{{\color{gray}#1}}

\hyphenation{Lip-schitz poly-no-mi-al-ly}

\begin{document}
\title{A Quantitative Framework for Navigating Controller Design Tradeoffs under Computational Constraints}
\author{%
Chris Verhoek and Nikolai Matni
\thanks{C.~Verhoek and N.~Matni are with the Dept. of Electrical and Systems Engineering, University of Pennsylvania, United States. C.~Verhoek is also with the Control Systems Group, Eindhoven University of Technology, The Netherlands. Email addresses: \texttt{\{cverhoek, nmatni\}@engineering.upenn.edu}, \texttt{c.verhoek@tue.nl}. Corresponding author: C.~Verhoek. This work has been supported by AFOSR (FA9550261B038) and the European project COVER under the grant No.~101086228.

Keywords: Optimal Control, Computational Limits, Tradeoff Analysis, Resource-aware Controller Design, Predictive control for linear systems, Lyapunov methods.}
}

\maketitle
\thispagestyle{plain}
\pagestyle{plain}

\begin{abstract}
	Computational constraints permeate the controller design process, and yet are rarely treated as explicit design constraints. Towards addressing this gap, we propose a quantitative framework that captures the effects of common design approximations, such as model order reduction, temporal discretization, horizon truncation, and solver accuracy, on both controller performance and computational requirements.  Our framework highlights that these approximations are \emph{tunable parameters} within an overall controller design process.  By leveraging incremental input-to-state stability, we show that bounding the aggregate effects of these approximations reduces to verifying a design-dependent sector bound on the difference between the deployed policy and an idealized baseline, with stability enforced via a small-gain condition. We operationalize these insights via a \emph{Design Meta-Problem} in which the performance gap is minimized subject to stability, real-time compute, and timing constraints. Finally, we instantiate the framework on a receding horizon LQR case study, and demonstrate a principled near-optimal navigation of tradeoffs among sampling rate, model order, horizon length, and solver iterations. 

\end{abstract}

\section{Conclusions and Outlook}\label{s:conclusion}\vspace{-0.2em}

We presented a framework for systematically navigating approximation choices in controller design under explicit computational constraints,  formalized as a Design Meta-Problem (DMP) that jointly optimizes controller performance subject to stability, real-time compute, and timing constraints. When the baseline system is $\delta$ISS, bounding closed-loop performance degradation reduces to verifying a design-dependent sector bound on the policy difference, with stability enforced via a small-gain condition. In the receding horizon LQR setting, this yields a scalar optimization over $\tau$ for each candidate model $M_i$, from which a jointly optimal $(\tau^\star, M^\star)$ is identified by direct comparison.

Looking ahead, we aim to extend the framework in three directions. 
First, for the LQR example, the sector bound gains $\hat L, L^\star$ were estimated by 
bounding errors in the feedback gains for various choices of $\tau$.  While this is tractable for linear feedback policies, it may not be in general, and developing general 
data-driven or empirical methods for fitting these gains, e.g., via sampling-based estimation of $\|(\hat{\pi} - 
\pi^\star)(x)\|$, is an important direction for broader applicability. 
Second, the DMP naturally accommodates additional 
approximations and constraints, such as memory constraints for explicit MPC or sector bound terms capturing model uncertainty from 
data-driven identification. 
Third, and more ambitiously, 
Theorem~\ref{thm:dmp} provides a principled building block for the 
systematic design of layered control 
architectures~\cite{matni2024quantitative}: when a single-layer DMP 
is infeasible, this motivates decomposing the problem into layers, 
each with its own compute resources and DMP, linked via inter-layer 
constraints.

\bibliographystyle{abbrvnat}
\bibliography{refs_cdc26tradeoff}

@article{chen2025sampled,
  title={Sampled-data systems: Stability, contractivity and single-iteration suboptimal {MPC}},
  author={Chen, Yiting and Bullo, Francesco and Dall'Anese, Emiliano},
  journal={arXiv preprint arXiv:2505.18336},
  year={2025}
}

@article{srikanthan2024closed,
  title={Closed-loop analysis of ADMM-based suboptimal linear model predictive control},
  author={Srikanthan, Anusha and Karapetyan, Aren and Kumar, Vijay and Matni, Nikolai},
  journal={IEEE Contr. Sys. Lett.},
  volume={8},
  pages={3195--3200},
  year={2024},
  publisher={IEEE}
}

@article{garg2021multi,
  title={Multi-rate control design under input constraints via fixed-time barrier functions},
  author={Garg, Kunal and Cosner, Ryan K and Rosolia, Ugo and Ames, Aaron D and Panagou, Dimitra},
  journal={IEEE Contr. Sys. Lett.},
  volume={6},
  pages={608--613},
  year={2021},
  publisher={IEEE}
}

@article{extended,
  title={A Quantitative Framework for Navigating Controller Design Tradeoffs under Computational Constraints},
  author={Verhoek, Chris and Matni, Nikolai},
  journal={arXiv preprint arXiv:2604.24897},
  year={2026},
  note={Extended version}}

@article{bottou2007tradeoffs,
  title={The tradeoffs of large scale learning},
  author={Bottou, L{\'e}on and Bousquet, Olivier},
  journal={Adv. Neur. Inf. Proc. Sys.},
  volume={20},
  year={2007}
}

@article{angeli2002lyapunov,
  title={A {Lyapunov} approach to incremental stability properties},
  author={Angeli, David},
  journal={IEEE Trans. Aut. Contr.},
  volume={47},
  number={3},
  pages={410--421},
  year={2002},
  publisher={IEEE}
}

@book{khalil2002nonlinear,
  title={Nonlinear Systems},
  author={Khalil, Hassan K},
  volume={3},
  year={2002},
  publisher={Prentice Hall}
}

@article{konstantinov1993perturbation,
  title={Perturbation analysis of the discrete {Riccati} equation},
  author={Konstantinov, Michail M and Petkov, P. {Hr.} and Christov, Nikolai D},
  journal={Kybernetika},
  volume={29},
  number={1},
  pages={18--29},
  year={1993},
  publisher={Institute of Information Theory and Automation AS CR}
}

@article{shitasos2025suboptimality,
  title={Suboptimality analysis of receding horizon quadratic control with unknown linear systems and its applications in learning-based control},
  author={Shi, Shengling and Tsiamis, Anastasios and de Schutter, Bart},
  journal={IEEE Trans. Aut. Contr.},
   year={2026},
  volume={71},
  number={3},
  pages={1422-1437}}

@incollection{alla2019order,
  title={Order reduction approaches for the algebraic {Riccati} equation and the {LQR} problem},
  author={Alla, Alessandro and Simoncini, Valeria},
  booktitle={Numerical Methods for Optimal Control Problems},
  pages={89--109},
  year={2019},
  publisher={Springer}
}

@incollection{bitmead1991riccati,
  title={Riccati difference and differential equations: Convergence, monotonicity and stability},
  author={Bitmead, Robert R and Gevers, Michel},
  booktitle={The Riccati Equation},
  pages={263--291},
  year={1991},
  publisher={Springer}
}

@article{hager1976convergence,
  title={Convergence and stability properties of the discrete {Riccati} operator equation and the associated optimal control and filtering problems},
  author={Hager, William W and Horowitz, Larry L},
  journal={SIAM J. Contr. \& Optim.},
  volume={14},
  number={2},
  pages={295--312},
  year={1976},
  publisher={SIAM}
}

@article{grune2008infinite,
  title={On the infinite horizon performance of receding horizon controllers},
  author={Gr{\"u}ne, Lars and Rantzer, Anders},
  journal={IEEE Trans. Aut. Contr.},
  volume={53},
  number={9},
  pages={2100--2111},
  year={2008},
  publisher={IEEE}
}

@article{matni2024quantitative,
	author = {Matni, Nikolai and Ames, Aaron D and Doyle, John C},
	journal = {IEEE Contr. Sys. Mag.},
	number = {3},
	publisher = {IEEE},
	title = {A quantitative framework for layered multirate control: Toward a theory of control architecture},
	volume = {44},
	year = {2024}}

@book{antoulas2005approximation,
  title={Approximation of large-scale dynamical systems},
  author={Antoulas, Athanasios C},
  year={2005},
  publisher={SIAM}
}

@book{chen1995optimal,
  title={Optimal sampled-data control systems},
  author={Chen, Tongwen and Francis, Bruce A},
  year={1995},
  publisher={Springer-Verlag}
}

@article{karapetyan2025closed,
  title={Closed-loop finite-time analysis of suboptimal online control},
  author={Karapetyan, Aren and Balta, Efe C and Iannelli, Andrea and Lygeros, John},
  journal={IEEE Trans. Aut. Contr.},
  year={2025},
  publisher={IEEE}
}

@article{liao2020time,
  title={Time-distributed optimization for real-time model predictive control: Stability, robustness, and constraint satisfaction},
  author={Liao-McPherson, Dominic and Nicotra, Marco M and Kolmanovsky, Ilya},
  journal={Automatica},
  year={2020},
  volume={117},
  pages={108973}
}

@inproceedings{zardini2021codesign,
  title={Co-design of autonomous systems: From hardware selection to control synthesis},
  author={Zardini, Gioele and Censi, Andrea and Frazzoli, Emilio},
  booktitle={Proc. 2021 Europ. Contr. Conf.},
  pages={682--689},
  year={2021}
}

@article{lahijanian2018resource,
  title={Resource-performance tradeoff analysis for mobile robots},
  author={Lahijanian, Morteza and Svorenova, Maria and Morye, Akshay A and Yeomans, Brian and Rao, Dushyant and Posner, Ingmar and Newman, Paul and Kress-Gazit, Hadas and Kwiatkowska, Marta},
  journal={IEEE Robot. Autom. Lett.},
  volume={3},
  number={3},
  pages={1840--1847},
  year={2018},
  publisher={IEEE}
}

@incollection{heemels2021event,
  title={Event-triggered and self-triggered control},
  author={Heemels, Wilhelmus P M H and Johansson, Karl H and Tabuada, Paulo},
  booktitle={Encycl. Sys. \& Contr.},
  year={2021},
  pages={724--730},
}

@article{rosolia2022unified,
  title={Unified multirate control: From low-level actuation to high-level planning},
  author={Rosolia, Ugo and Singletary, Andrew and Ames, Aaron D},
  journal={IEEE Trans. Aut. Contr.},
  volume={67},
  number={12},
  pages={6627--6640},
  year={2022},
  publisher={IEEE}
}

@inproceedings{stamouli2025layered,
  title={Layered multirate control of constrained linear systems},
  author={Stamouli, Charis and Tsiamis, Anastasios and Morari, Manfred and Pappas, George J},
  booktitle={Proc. 64th IEEE Conf. Decis. Contr.},
  pages={3027--3034},
  year={2025}
}

@article{richter2011computational,
  title={Computational complexity certification for real-time MPC with input constraints based on the fast gradient method},
  author={Richter, Stefan and Jones, Colin Neil and Morari, Manfred},
  journal={IEEE Trans. Aut. Contr.},
  volume={57},
  number={6},
  year={2011},
  publisher={IEEE}
}

@article{rubagotti2014stabilizing,
  title={Stabilizing linear model predictive control under inexact numerical optimization},
  author={Rubagotti, Matteo and Patrinos, Panagiotis and Bemporad, Alberto},
  journal={IEEE Trans. Aut. Contr.},
  volume={59},
  number={6},
  pages={1660--1666},
  year={2014},
  publisher={IEEE}
}

@article{kenney1990sensitivity,
  title={The sensitivity of the algebraic and differential Riccati equations},
  author={Kenney, Charles and Hewer, Gary},
  journal={SIAM J. Contr. Optim.},
  volume={28},
  number={1},
  pages={50--69},
  year={1990},
  publisher={SIAM}
}

@article{bonnans2017error,
  title={Error estimates for the {Euler} discretization of an optimal control problem with first-order state constraints},
  author={Bonnans, J Fr{\'e}d{\'e}ric and Festa, Adriano},
  journal={SIAM J. Numerical Analysis},
  volume={55},
  number={2},
  pages={445--471},
  year={2017},
  publisher={SIAM}
}

\appendix

\end{document}

\fi